\documentclass[12pt]{article}
\usepackage[T1]{fontenc}
\usepackage{geometry}
\usepackage{natbib}
\usepackage{float}
\usepackage{hyperref}
\usepackage{arydshln}
\usepackage{bm}
\usepackage{tikz}
\usepackage{bbm}
\usepackage{amsmath}
\usepackage{amsthm}
\usepackage{amssymb}
\usepackage{booktabs,caption}
\usepackage{threeparttable}
\usepackage{graphicx}
\usepackage{color}
\makeatletter
\theoremstyle{plain}

\theoremstyle{plain}
\newtheorem{assumption}{\protect\assumptionname}[section]
\theoremstyle{plain}
\newtheorem{prop}{\protect\propositionname}
\theoremstyle{plain}
\newtheorem{lem}{\protect\lemmaname}[section]
\theoremstyle{plain}

\newtheorem{remark}{Remark}[section]

\newtheorem{definition}{Definition}
\makeatother
\usepackage{authblk}
\usepackage[english]{babel}
\providecommand{\assumptionname}{Assumption}
\providecommand{\lemmaname}{Lemma}
\providecommand{\propositionname}{Proposition}
\providecommand{\theoremname}{Theorem}

\numberwithin{equation}{section}
\begin{document}
	
	\title{Identification and Estimation of A Rational Inattention Discrete Choice Model with Bayesian Persuasion}
	\author{Moyu Liao\thanks{First draft: 03/22/2019. This version: 09/16/2020. I would like to thank Marc Henry, Sun Jae Jun, Peter Newberry, Karl Schurter, Jia Xiang, Zhiyuan Chen and conference attenders at the 2019 Econometric Society Asian Meeting for useful comments. }}
	\affil{Pennsylvania State University}
	\maketitle
	
	\begin{abstract}
		This paper studies the semi-parametric identification and estimation of a rational inattention model with Bayesian persuasion. The identification requires the observation of a cross-section of market-level outcomes. The empirical content of the model can be characterized by three moment conditions. A two-step estimation procedure is proposed to avoid computation complexity in the structural model. In the empirical application, I study the persuasion effect of Fox News in the 2000 presidential election. Welfare analysis shows that persuasion will not influence voters with high school education but will generate higher dispersion in the welfare of voters with a partial college education and decrease the dispersion in the welfare of voters with a bachelors degree. 
	\end{abstract}
	
	\pagebreak

	\section{Introduction}
	In many applications of discrete choice models, econometricians usually assume the decision maker has the following random utility from choosing item $j$ among a choice set $\mathcal{J}={1,...J}$: $U_{j}=u_{j}+\epsilon_{j}$,
	where $u_{j}$ is the mean utility observed by the econometricians and $\epsilon_{j}$ is the utility shock known to the decision maker but not the econometrician. Decision makers in the model choose the item with the highest utility. When the unobserved shock follows the Type I extreme value distribution, we can solve the probability of choosing $j$ analytically. Aggregating the choice outcomes of the decision makers in the market we can get the market share of an item. This approach to studying market structure was initiated by \citet*{mcfadden1973}, and then adopted by  \citet*{Berry1995} (henceforth BLP) to study automobile markets, and became widely applied to other industries.
	
	This model, however, is not easily adaptable to accommodate persuasion in a structural way. Take advertising as a form of persuasion. In the classical analysis of the effect of advertising, three approaches are adopted. The first is to model advertising as a feature of the item that enters mean utility $u_{j}=u_{j}(A)$, where the level of advertising $A$ affects the choice utility. The argument is that advertising is `persuasive' and the individual will buy more of the advertised goods because their utility is distorted \citep*{Dorfman1954}. This reduced form approach does not offer us much explanation of how advertisement influences decision making and market structure. The second approach is to model advertisement as the trigger of the consideration set change \citep*{Goeree2008}. The consideration set is the priori set of items that the decision maker chooses from. Advertisement thus serves as the trigger that puts a previously non-considered item into the consideration set. This approach views advertisement as the information revealing device that reveals the true $\epsilon_{j}$ to the decision maker which was previously $-\infty$ to the decision maker. If a good $j$ is already in the consideration set for all customers, the model of consideration set predicts that advertising has no effect on the market share. If a good is well known, the model of consideration set cannot explain why sellers advertise. The third approach is to view advertising as a signaling device to separate the high-quality product from the low-quality product \citep*{Nelson1974,Bagwell1988}. The degree of advertisement serves as the signal that induces the separating equilibrium where only high-quality firms advertise. In particular, they assume the unobserved quality is common for all decision makers. However, this approach requires the decision maker in the model to have imperfect knowledge of $\epsilon_{j}$, which contradicts the assumption that $\epsilon_j$ is known by the decision maker.   
	
	Compared to the classical approaches to model persuasion, this paper develops an empirical model of persuasion using the Bayesian persuasion theory in  \citet{Kamenica2011}. The Bayesian persuasion approach to model advertising differs from the previously mentioned informative view in two ways: 1. Decision makers in the model can have a different realization of the product quality; 2. The advertiser, who acts as the Bayesian persuader, does not always want to reveal their quality honestly. However, similar to the informative view, the Bayesian persuasion model assumes that the decision maker in the model only has a prior belief on the $\{\epsilon_{j}\}_{j=1}^J$ and the exact realization of $\{\epsilon_{j}\}_{j=1}^J$ is unknown. The decision maker's prior distribution of $\{\epsilon_{j}\}_{j=1}^J$ comes from the reputation of the goods. The prior belief is likely to be common across decision makers. However, the standard Bayesian persuasion model assumes the decision makers only have access to the signal sent by the persuader to update their belief and no other sources of information are available. In the real world, decision makers will also search actively for information on the goods' quality by themselves. For example, if a person wants to buy a car, he or she will have a test drive before making a  decision. An extensive search of information can reduce the randomness of  $\{\epsilon_{j}\}_{j=1}^J$ but at the time is costly. \citet*{Matejka2015} considers a model where the decision maker searches information of $\{\epsilon_{j}\}_{j=1}^J$ to maximize the expected utility after deducting the search cost. Their rational inattention discrete choice model can incorporate Bayesian persuasion by assuming persuaders send signals after the decision makers get their own information. 
	
	The analysis of the structural persuasion and information search model has largely been discussed under the assumption that the decision makers' prior belief of $\{\epsilon_{j}\}_{j=1}^J$, denoted by $G$, is known to the economist. In empirical researches, the prior belief $G$ is unknown and should be estimated from data. A recent empirical study by \citet{XiangjiaJMP} assumes the decision makers' prior distribution $G$ is normally distributed and analyzes the decision makers' welfare change when a policy change induces the persuader to change the persuasion strategy. However, the empirical content of a parametric assumption on $G$ is unclear.
	
	
	This paper follows \citet*{Matejka2015} to consider a rational inattention discrete model with Bayesian persuasion. I discuss the non-parametric identification of the prior distribution $G$ and parametric identification of persuader's persuasion strategy when an econometrician observes the choice ratio at the market level across many independent markets. The independent markets are divided into two groups: the first group is not influenced by the persuader and the second group is influenced by the persuader. The prior distribution  $G$ is identified from the choice ratio in the first group of markets.  Given the identification of $G$, a parametric persuasion strategy is identified from the second group of markets. I characterize a set of moment conditions implied by the model,  and the standard estimation method such as GMM can be applied easily. 
	
	For econometricians who already observe the market shares with and without the influence of persuasion, identifying the persuasion strategy is the first step to understanding the behavior of the persuader. If we assume the persuader use a persuasion strategy is to maximize some utility function, the identified persuasion strategy can help us understand the persuader's objective function. Analysis in this paper leaves the persuader's objective function as unknown and analyzes the behavior from the buyers' side. A complete two-sided analysis will incorporate the persuader's utility as a function of persuasion strategy and analyze the problem as a sequential game played between the persuader and the buyers.
	
	For policymakers, given the knowledge of the prior belief $G$, they will be able to evaluate the effect of regulating the persuasion strategy. In the advertisement market, the policymakers for example can ban one seller from directly revealing information about his competitors' products. Moreover, policymakers can also evaluate the effect of providing less costly information to the decision makers. In other words, policymakers can compete with existing persuaders in the markets to increase the decision makers' welfare. 
	
	In the empirical application, I look at the 2000 presidential election in the United States. I treat the presidential candidates as voters' choices and view voting statistical areas as separated markets. In 1996, Fox News was developed and then entered into approximately 30\% of the towns in the United States by 2000. \citet*{DellaVigna2007} shows that Fox News motivated voters to vote for Republicans compared to voters in towns without Fox News. I take the data and analyze how Fox News persuaded voters in different towns. The estimated results from the markets without Fox News show that the prior belief of the quality of the presidential candidates varies a lot with voters' education level. Both voters with bachelor's degrees and with only high school degrees prefer the Democratic party than the Republican Party. The estimated results also show that Fox News provided very little information to voters, but managed to manipulate the voting outcome by a significant margin. I also compare the welfare of voters with different education levels. Voters' welfare is defined as the probability of choosing their first best choice, and their first best choice is the presidential candidate that will generate the highest utility to voters when the voters know the realization of $\{\epsilon_{j}\}_{j=1}^J$. The result shows that persuasion will not influence the welfare of voters with high school education but will generate higher dispersion in the welfare of voters with a partial college education and decrease the dispersion in the welfare of voters with a bachelors degree.

	Another way to study the effect of persuasion is to model the presence of a persuader as a treatment status \citep*{Jun2018}. In their model, the presence of a Bayesian persuader is taken as treatment assignment and sharp bounds on the persuasion effect are given under various data generating processes. The treatment effect model does not specify the decision makers' utility and thus analysis of the decision makers' welfare before and after persuasion is not possible. The treatment effect model also makes it hard to consider policy counterfactual such as regulations on persuasion strategy or when the policymaker provides extra information in the market. 
	
	The rest of the paper is organized as follows. Section 2 introduces the rational inattention discrete choice model with persuasion. Section 3 discusses the data generating process and the identification strategy. Section 4 discusses the estimation strategy. Section 5 studies the 2000 presidential election and the effect of Fox News. Section 6 concludes.

	\section{The Model}
	I consider the standard random utility specification: a decision maker (DM) derives utility level $U_{j}$ from good $j$ from the choice set $\mathcal{J}=\{1,..., J\}$:
	\[
	U_{j}=u_{j}+\epsilon_{j}.
	\]
	The $u_j$ is the mean utility of choosing good $j$ and $\epsilon_{j}$ is the individual specific random draw of utility shock. Throughout this section, I assume that the decision maker knows only $(u_1,...,u_j)$ but not $(\epsilon_{1},...\epsilon_{J})$. The decision maker has a prior belief on the distribution $G$ on the utility shock: $(\epsilon_{1},...\epsilon_{J})\equiv \bm{\epsilon}\sim G$. If there is no further information about the true utility shock $\bm{\epsilon}$, the decision maker will choose the one with highest expected utility:
	\begin{equation}\label{eq: choice criteria}
	j\in a(G)\equiv \arg\max_{j\in\mathcal{J}} E_{G}[u_{j}+\epsilon_{j}].
	\end{equation}
	If $\arg\max_{j\in\mathcal{J}} E_{G}[u_{j}+\epsilon_{j}]$ is not a singleton, we let $a(G)$ to be an arbitrary selection of maximizers. The maximized utility derived from the belief $G$ is given by 
	\begin{equation}\label{choice utility}
	V(G)\equiv \max_{j\in\mathcal{J}} E_{G}[u_{ij}+\epsilon_{ij}].
	\end{equation}
	
	I will first introduce a rational inattention discrete choice model and then discuss how persuasion can be incorporated.

	\subsection{Rational Inattention Discrete Choice Model}
	The rational inattention discrete choice model in \citet{Matejka2015} assumes that the decision maker can choose an information strategy to get a signal $\mathbf{s}^{DM}$. The signal $\mathbf{s}^{DM}$ updates the decision makers' belief on the true utility shock $\bm{\epsilon}$. The decision maker then choose the item with highest posterior mean according to (\ref{eq: choice criteria}). Following the notation in \citet{Matejka2015}, denote $u_j+\epsilon_j\equiv v_j$. Formally, the decision maker's information strategy is a joint distribution of the true utility vector $\mathbf{v}\in \mathbb{R}^J$ and the signal $\mathbf{s}^{DM}\in \mathbb{R}^J$, denoted by $F(\mathbf{s}^{DM},\mathbf{v})$. The marginal distribution of the information strategy has to be consistent with the prior belief $G$. Once the decision maker is committed to the information strategy, the random shocks to utility are realized, and then the decision maker get a realized signal $\mathbf{s}^{DM}$ from $F(\mathbf{s}^{DM}|\mathbf{v})$. The decision maker updates his belief as $F(\bm{\epsilon}|\mathbf{s}^{DM})$, and chooses the item in $a(F(\bm{\epsilon}|\mathbf{s}^{DM}))$ according to (\ref{eq: choice criteria}). 
	
	Since the real utility shocks are not observed by the decision maker, the decision maker solves the following optimization problem to maximize his expected utility:
	\begin{equation}\label{eq: RI optimization}
	\max_{F\in \Delta({\mathcal{R}^{2J}}) } \int_{\mathbf{v}}\int_{\mathbf{s}^{DM}} V(F(\cdot|s^{DM}))F(d\mathbf{s}^{DM}|\mathbf{v})G(d\mathbf{v})-c(F)
	\end{equation}
	\begin{equation}\label{RI constraint}
	s.t. \, \int_{\mathbf{s}^{DM}}F(d\mathbf{s}^{DM},\mathbf{v})=G(\mathbf{v})
	\end{equation}
	where $V(F(\cdot|\mathbf{s}^{DM}))$ is determined by (\ref{choice utility}). The constraint (\ref{RI constraint}) requires that the DM's prior distribution $G$ is consistent with the real state of the world. The cost of information $c(F)$ is the  mutual information between the shocks $\bm{\epsilon}$ and the signal $s^{DM}$:
	\begin{equation}\label{mutual information}
	c(F)=\lambda \{H(G)-E_{\mathbf{s}}[H(F(\cdot|\mathbf{s}^{DM}))]\},
	\end{equation}
	where the parameter $\lambda$ is the unit cost of information, and $E_s$ denote the expectation over the marginal distribution of $F(\mathbf{s}^{DM},\mathbf{v})$. The entropy function $H$ of a discrete distribution $G$ is defined as:$H(G)=-\sum_k P_k \log(P_k)$,  
	where $P_k$ is the probability of the state $k$. When $G$ is continuously distributed, the differential entropy is defined as $H(G)=-\int_s g(s) \log(g(s)) ds$. The use of entropy reduction as a measure of information cost is standard in the rational inattention literature. See \citet{de2017rationally} for the discussion of entropy cost. Moreover, the entropy number is related to the complexity of a random variable, and can be given a data compression interpretation. The mutual information in (\ref{mutual information}) can be interpreted as the number of binary questions asked by acquiring signal $\mathbf{s}$. Appendix \ref{section: Data Compression and Entropy Cost} gives an example of data compression interpretation.
	
	Let $S_j^{DM} \equiv \{\mathbf{s}^{DM}\in\mathcal{R}^J: a(F(\cdot|\mathbf{s}^{DM})=j\}$ be the set of signals that lead the DM to choose $j$. Also denote 
	\begin{equation}\label{eq: conditional choice probability}
	\mathcal{P}_j(\mathbf{v})\equiv\int_{S_j^{DM}} F(d\mathbf{s}^{DM}|\mathbf{v})
	\end{equation}
	as the conditional choice probability of choosing item $j$ when the realized utility vector is $\mathbf{v}$ \footnote{Note that the DM does not know the realization of $\mathbf{v}$. The conditional choice probability should be understood to be the choice probability when the actual utility vector is $\mathbf{v}$}. Also define the unconditional choice probability of choosing $j$ as
	\begin{equation} \label{eq: unconditional choice probability}
	\mathcal{P}_j^0=\int_{\mathbf{v}} \mathcal{P}_j(\mathbf{v}) dG(\mathbf{v}).
	\end{equation}
	This is the ex-ante probability of choosing $j$ before the utility vector is realized.

	A set of optimality condition to the problem (\ref{eq: RI optimization})-(\ref{mutual information}) from \citet{Matejka2015} is  summarized in the following lemma.
	\begin{lem} \label{lemma: Matejka and Mckay}
		If $\lambda$>0 and $F$ is an optimal information strategy that solves  (\ref{eq: RI optimization})-(\ref{mutual information}), then the conditional and choice probability in (\ref{eq: conditional choice probability})  satisfies
		\begin{equation} \label{eq: Logit Form}
		\mathcal{P}_j(\mathbf{v})=\frac{\mathcal{P}^0_j {e^{v_j/\lambda}}}{\sum_{k\in\mathcal{J}}\mathcal{P}^0_k {e^{v_k/\lambda}}} \,\, a.s.,
		\end{equation}
		\begin{equation} \label{eq: con P inte to unc}
		E_G[\mathcal{P}_j(\mathbf{v})]=\mathcal{P}_j^0.
		\end{equation}
		The unconditional choice probability in (\ref{eq: unconditional choice probability}) solves the following convex optimization problem:
		\begin{equation} \label{eq: alterantive optimization}
		\begin{split}
		\max_{\{\mathcal{P}^0_j\}_{j=1}^J} \int_{\mathbf{v}} &\lambda \log(\sum_{j=1}^J \mathcal{P}^0_j e^{v_j/\lambda}) G(d\mathbf{v})\\
		&s.t. \, \, \forall j: \mathcal{P}_j^0\ge 0,\\
		&\sum_{k=1}^J \mathcal{P}_k^0 =1.
		\end{split}
		\end{equation}
		
		Conversely, if $\{\mathcal{P}_j^0\}_{j=1}^J$ is the solution to (\ref{eq: alterantive optimization}), and $\mathcal{P}_j(\mathbf{v})$  defined in (\ref{eq: Logit Form}) satisfies (\ref{eq: con P inte to unc}), then we can construct an information strategy $F$ such that:
		\begin{itemize}
			\item The signal $\mathbf{s}^{DM}$ is supported on $J$ points: $\{s_1,...s_J\}$;
			\item The conditional distribution of $\mathbf{s}^{DM}$ satisfies $Pr_F(\mathbf{s}^{DM}=s_j)= \mathcal{P}_j(\mathbf{v})$.
		\end{itemize}
		This information strategy $F$ solves the optimization problem (\ref{eq: RI optimization})-(\ref{mutual information}).
	\end{lem}
	\begin{proof}
		See Theorem 1 and Lemma 2 in \citet{Matejka2015}.
	\end{proof} 
	
	Lemma \ref{lemma: Matejka and Mckay} shows that solving the optimization problem (\ref{eq: RI optimization})-(\ref{mutual information}) is equivalent to solve the optimization problem (\ref{eq: alterantive optimization}). We do not observe the DM's optimal information strategy. Instead, we observe their choice outcome. When we aggregate the choice outcome to the market level, it becomes the conditional and unconditional choice probability.
	
	We should note that the conditional choice probability (\ref{eq: Logit Form}) takes a Logit-like choice probability form. However, the rational inattention discrete choice model does not imply the usual I.I.A constraints on the choice probability. \citet{Matejka2015} discusses the two equivalent  conditions on the conditional choice probability (\ref{eq: Logit Form}).
	
	\subsection{A Sequential Persuasion Game}
	Consider a persuader that tries to influence the choice probability by choosing a persuasion strategy and sending a realized signal. The persuader is also called he information designer ({ID}) in the Bayesian persuasion literature.
	\begin{definition}\label{def: persuasion strategy}
		A persuasion strategy is a joint distribution $\tilde{F}(s^{ID},\mathbf{v})$ of the signal $s^{ID}\in \mathcal{R}^J$ sent by the ID and the utility vector such that
		\[\int \tilde{F}(s^{ID},\mathbf{v})ds^{ID}=G(\mathbf{v}).\]
	\end{definition}
	
	I consider a sequential persuasion game between the decision makers and the information designer in the following order
	\begin{enumerate}
		\item The information designer chooses an persuasion strategy and then sends the realized signal $s^{ID}$ to the decision maker;
		\item The decision maker updates his belief to the intermediate distribution:
		\begin{equation}\label{update belief}
		\tilde{G}_{s^{ID}}\equiv \tilde{G}(\mathbf{v}|S^{ID}=s^{ID})= \frac{G(\mathbf{v})\times \tilde{F}(s^{ID}|\mathbf{v})}{\int_{\mathbf{v}} \tilde{F}(\mathbf{v},s)d\mathbf{v}};
		\end{equation}
		\item The decision maker solves optimization (\ref{eq: RI optimization})-(\ref{mutual information}) with the intermediate belief $\tilde{G}_{s^{ID}}$;
		\item The decision maker gets a realized signal $s^{DM}$ from his optimal information strategy $F$. He then makes the choice based on the updated belief $F(\mathbf{v}|s^{DM})$.
	\end{enumerate}
	
	For the persuasion strategy to work, it is assumed that the DM who receives the signal knows the joint distribution $\tilde{F}(s^{ID},\mathbf{v})$. 
	\begin{assumption} \label{common knowledge on persuation strategy}
		The persuasion strategy $\tilde{F}(s^{ID},\mathbf{v})$ is common knowledge.
	\end{assumption}
	
	The assumption \ref{common knowledge on persuation strategy} on $\tilde{F}$ is satisfied when there is an underlying equilibrium determining how the information designer chooses the persuasion strategy. For example, information designer can have an objective function $M:\Delta(\mathcal{R}^{2J})\rightarrow \mathcal{R}$  so that $\tilde{F}=\arg\max_{F\in C} M(F)$ where $C\subset \Delta(\mathcal{R}^{2J})$ is some constrained persuasion strategy set. When the objective function $M$ and the constrained set $C$ is known by the DM, the decision maker can solve the information designer's optimization problem to get $\tilde{F}$. This paper does not tackle the information designer's objective function. The objective function for the information designer is not easy to formulate. In the marketing context, the trade-off is between higher marketing cost of persuasion and higher sales. In the context of political persuasion, the goal of persuasion is not to maximize voting share but to increase voting share until it exceeds 50\%. Also, the media that conducts persuasion may also care about other aspects of persuasion since their persuasion strategy can influence their audience ratings.

	%
	%
	%
	%
	%
	
	The setting of the persuasion game is different from the setting in \citet*{bloedel2018persuasion}. In their setting, the decision maker chooses an information strategy to understand the signal send by the sender. In other words, the decision maker in their model pays attention cost to understand the signal from the sender and cannot acquire a signal about the true utility by himself. In my formulate, there is no cost to understand the signal $s^{ID}$ from the sender and there is a cost incurred by acquiring information about the true utility vector. 
	\begin{remark}
		The persuasion strategy $\tilde{F}$ and the information strategy $F$ lie in the same space. The effect of persuasion is limited because decision makers can acquire their own information. While the ID can distort the prior distribution of the utility vector $\mathbf{v}$ through $\tilde{F}$, the decision maker's information strategy $F(\mathbf{s}^{DM},\mathbf{v})$ can still provide information to the decision maker. 
	\end{remark}

	\section{Identification}
	In this section, I discuss a data generating process that allows us to non-parametrically identify the prior belief $G$ and parametrically identify the persuasion strategy $\tilde{F}$. To allow for the heterogeneity of decision makers' preferences, I assume the utility of an individual $i$ in market $m$ in the demographic group $k$ takes the following additively separable form:
	\begin{equation} \label{eq: parameteric utility assumption}
	U_{ikjm}=u_1(x^m_{j},\beta)+u_2(x^m_{j},\nu_{ik},\alpha)+\epsilon_{j,m},
	\end{equation}
	where $x^m_j$ is the characteristics of product $j$ in market $m$; $k$ is the index for people of demographic group $k$ with demographic characteristics $\nu_{ik}$; $m$ is the index for market. The utility function $u_1,u_2$ is of known parametric form, and $\alpha$, $\beta$ are two vectors to be estimated, but the distribution of $G$ is left as non-parametric. The utility (\ref{eq: parameteric utility assumption})  assumes that the decision makers'  demographic and product characteristics only influence their mean utility but not the utility shocks.  Here I assume that all DMs in the same market $m$ realize the same $\bm{\epsilon}_{m}=(\epsilon_{1m},...,\epsilon_{jm})$ since the random shock vector $\bm{\epsilon}_{m}$ in equation (\ref{eq: parameteric utility assumption}) does  not depend on the individual index $i$. In particular, if individual $i_1$ and $i_2$ are in the same market, and they have the same demographic characteristics, they should have the same realized utility vector. This specification is reasonable when the shock is market-specific. For example, when we want to study the voting decision, the market realization of $\bm{\epsilon}_{m}$ can be the real payoff of candidate $j$'s policy on town $m$'s local industry. In the automobile industry, this market level states of the world may come from the local road condition, climate, or geographic topology. 
	
	
	\subsubsection*{Notation}
	Throughout this section, I use $\tilde{\cdot}$ to denote probability quantities related to markets with the presence of a persuader. Also, I drop the super-scrip on $s^{ID}$, and use $s$ to denote signals sent by the information designer whenever there is no confusion. I use $m$ to denote the index for markets, $j$ to denote the index of products, $k$ to denote the index of different demographic groups.
	
	\subsection{Data Generating Process}
	In many data sets, we do not observe individual choices. Instead, we observe the market share, which is the aggregated individual choices. Across different markets, I assume that the prior distribution on $\epsilon_{j,m}$ is the same $G$.

	\begin{assumption}\label{assum: Data}
		({Data}) (i) We observe a binary variable $\chi^m$ such that $\chi^m=1$ if and only if the persuader is present in market $m$; (ii). The demographic heterogeneity $v_k$ is discrete and supported on $K$ points. For each market $m$, the distribution of demographic heterogeneity $D^m=(d^m_1,...,d^m_K)$ is observed, where $d_k^m$ is the proportion of DMs in group $k$ in market $m$; (iii). We observe the market characteristics $\mathbf{X}^m$ in each market $m$ and the market share vector $\bm{ms}^m=({ms}^m_1,...,{ms}^m_J)$, where ${ms}^m_j$ is the market share of product $j$.
	\end{assumption}
	
	The following is the assumption on the markets where there is no persuader. 
	\begin{assumption}\label{assum: (DGP without Persuasion) }
		(DGP without Persuasion) For markets with $\chi^m=0$, the data generating process satisfies: 
		\begin{enumerate}
			\item Common prior: $\bm{\epsilon^m}\sim G$;
			\item Independent random utility shock: $\bm{\epsilon}^m\perp \mathbf{X}^m$
			\item Independent demographic distribution: $D^m\perp (\bm{\epsilon}^m,\mathbf{X}^m)$
			\item  The choice set $\mathcal{J}$ and information cost $\lambda$ are the same across different markets.
		\end{enumerate}
	\end{assumption} 
	Assumption \ref{assum: (DGP without Persuasion) } imposes that the mean of $\bm{\epsilon^m}$ is independent of the product characteristics and is normalized to be zero. If there is any unobserved characteristics that is correlated with $X_j^m$, the unobserved effects are captured by the observed $X_j^m$.  
	
	For markets with a persuader, I assume that the persuader is the same across these markets and the persuader use the same persuasion strategy.  Moreover, I assume that the persuasion strategy is a joint distribution of $\bm{\epsilon}^m$ and $s^{ID}$. This specification is different from Definition \ref{def: persuasion strategy} and the persuader uses the same persuasion strategy even if the product characteristics $X_j^m$ may vary across markets. 
	\begin{assumption} \label{assum: DGP with Persuasion} (DGP with Persuasion)
		For markets with $\chi^m=1$, the data generating process satisfies: 
		\begin{enumerate}
			\item $(\bm{\epsilon},X_j^m,D^m)$ and $(\lambda,\mathcal{J})$ satisfy the conditions in Assumption \ref{assum: (DGP without Persuasion) };
			\item There is a uniform persuader across markets with $\chi^m=1$ and
			the persuasion strategy $\tilde{F}^k(s^{ID},\bm{\epsilon})$ can depend on the demographic groups $k$ but not the market;
			\item The persuasion signal $s^{ID}_{km}\sim \tilde{F}^k(s^{ID}|\bm{\epsilon}_m)$ and the signal $s^{ID}_{km}$ is independent of each other across demographic groups and markets;
			\item Signal Independence $(s^{ID}_k,\bm{\epsilon}^m)\perp D^m$.
		\end{enumerate}
	\end{assumption}

	\subsubsection*{Normalization}
	Since the permutation of the item index does not matter, I call the last item $J$ the outside option. Note that in the discrete choice model, only the relative difference of utility matters for the DM.  Therefore, we can normalize the utility of outside option $J$ to be zero $U_{kJ}=0$. Also, when $u_1,u_2$ is homogeneous of degree one with respect to $\alpha,\beta$, the vector $(\alpha,\beta,\lambda,G)$ is not identified. Indeed, we can consider a model with $(c\alpha,c\beta,c\lambda,cG)$, where $cG$ is the distribution of $c\bm{\epsilon}_m$. The model $(c\alpha,c\beta,c\lambda,cG)$ will generate the same choice probability (\ref{eq: conditional choice probability} and (\ref{eq: unconditional choice probability}). Since linear specification of utility is frequently used in the applied literature, I assume $u_1,u_2$ is homogeneous of degree one with respect to $\alpha,\beta$.
	
	\begin{assumption}(Normalization)\label{assum: normalization}
		The utility functions $u_1,u_2$ are homogeneous of degree 1 with respect to  $(\alpha,\beta)$, and $\lambda=1$.
	\end{assumption}  
	
	\subsection{The Identified Set}
	The parameters of interests include the mean utility parameters $(\alpha,\beta)$, the prior belief $G$, and the persuasion strategy $\tilde{F}$. For markets without persuasion, we are also interested in $\mathcal{P}_j^{0,k}(\mathbf{X})$, which is demographic group $k$'s unconditional choice probability of choosing $j$ when the product characteristics are $\mathbf{X}$. If we want to evaluate the overall effect of persuasion across different markets, we want to compare the post-persuasion market share with $\mathcal{P}_j^{0,k}(\mathbf{X})$.
	
	I first define the identified set of $(\alpha,\beta,G,\mathcal{P}_j^{0,k}(\mathbf{X}))$ from the rational inattention discrete choice model. 
	\begin{definition}\label{def: identified set of the RIDC model}
		Let $\mathcal{F}_{\chi=0}$ denote the conditional distribution of $(D^m,\mathbf{X}^m,\bm{ms}^m)$ conditioned on $\chi^m=0$. The identified set of $(\alpha,\beta,G,\mathcal{P}_j^{0,k}(\mathbf{X}))$ under the rational inattention discrete choice model, denoted by $\Gamma_I$, is the collection of $(\alpha,\beta,\{\mathcal{P}_j^{0,k}(\mathbf{X})\}_{j,k},G)$ that satisfies the following constraints:
		\begin{enumerate}
			\item Given $(\alpha,\beta,G)$, $\{\mathcal{P}_j^{0,k}(\mathbf{X})\}_{j,k}$ solves the individuals optimization problem (\ref{eq: alterantive optimization}) with \begin{equation}\label{eq: the v_j with linear specification}
			v^m_j=u_1(x^m_{j},\beta)+u_2(x^m_{j},\nu_{ik},\alpha)+\epsilon^m_j;
			\end{equation}
			
			\item The unconditional mean of the conditional choice probability is the unconditional choice probability:
			\begin{equation}\label{eq: mean of con is unc}
			E_G\left[\frac{\mathcal{P}_j^{0,k}(\mathbf{X}) {e^{v^m_j/\lambda}}}{\sum_{k\in\mathcal{J}}\mathcal{P}_j^{0,k}(\mathbf{X}) {e^{v^m_k/\lambda}}}\right] = \mathcal{P}_j^{0,k}(\mathbf{X});
			\end{equation}
			
			\item Consider the mapping:
			\begin{equation}\label{eq: market share is the weighted sum of group share}
			\mathcal{P}_j^{m}(\alpha,\beta,\bm{\epsilon},D^m,\mathbf{X}^m,\{\mathcal{P}_j^{0,k}(\mathbf{X})\}_{j,k}) = \sum_{k} d_k^m \frac{\mathcal{P}_j^{0,k}(\mathbf{X}) {e^{v^m_j/\lambda}}}{\sum_{k\in\mathcal{J}}\mathcal{P}_j^{0,k}(\mathbf{X}) {e^{v^m_k/\lambda}}}
			\end{equation}
			where $v_j^m$ is defined in (\ref{eq: the v_j with linear specification}). Then $(D^m,X_j^m,\mathcal{P}_j^{m}(\alpha,\beta,\bm{\epsilon},D^m,\mathbf{X}^m,\{\mathcal{P}_j^{0,k}(\mathbf{X})\}_{j,k}))$ has the same distribution as $\mathcal{F}_{\chi=0}$.
			
		\end{enumerate}
	\end{definition}
	
	The first two conditions in Definition \ref{def: identified set of the RIDC model} corresponds to the optimization condition (\ref{eq: alterantive optimization}) and the condition (\ref{eq: con P inte to unc}) in Lemma \ref{lemma: Matejka and Mckay}. Equation (\ref{eq: market share is the weighted sum of group share}) calculates the market share of product $j$ as the weighted average of different demographic groups' choice probability. The third condition in Definition \ref{def: identified set of the RIDC model} requires the model predicted market share is consistent with the observed data distribution.
	
	I then define the identified set of the persuasion strategy $\tilde{F}$.
	\begin{definition}\label{def: identified set of the persuasion strategy}
		Let $\mathcal{F}_{\chi=1}$ denote the conditional distribution of $(D^m,\mathbf{X}^m,\bm{ms}^m)$ conditioned on $\chi^m=1$. Given the value of $(\alpha,\beta,G)$, and a persuasion strategy $\tilde{F}$, we consider the map:
		\begin{equation}\label{eq: condi choice prob with persuasion}
		\tilde{\mathcal{P}}_{j,s}^{k}(\bm{\epsilon};\mathbf{X}^m)=\frac{\tilde{\mathcal{P}}_{j,s}^{0,k}(\mathbf{X}^m)e^{v_j^m}}{\sum_{l=1}^J\tilde{P}_{l,s}^{0,k}(\mathbf{X}^m)e^{v_j^m}}
		\end{equation}
		where $\tilde{\mathcal{P}}_{j,s}^{0,k}(\mathbf{X}^m)$ solves the individual optimization problem (\ref{eq: alterantive optimization}) when his belief is $\tilde{G}(\bm{\epsilon})=\tilde{F}(\bm{\epsilon}|s^{ID}=s)$.
		The identified set of the persuasion strategy is the set of $\tilde{F}(s^{ID},\bm{\epsilon})$ such that \[\left(D^m,\mathbf{X}^m,  \sum_k d_k^m \tilde{\mathcal{P}}_{j,s}^{k}(\bm{\epsilon};\mathbf{X}^m) \right)\]
		has the same distribution as $\mathcal{F}_{\chi=1}$.	
	\end{definition}
	
	The identified set in Definition \ref{def: identified set of the persuasion strategy} is conditioned on the vector $(\alpha,\beta,G)$. This is because in markets with the persuader, there are two types of the unobserved heterogeneity: the utility shock $\bm{\epsilon}^m$ and the realization of the signal $s^{ID}$. In contrast, in markets without the persuader, only the utility shock $\bm{\epsilon}^m$ exists. Therefore, knowing the prior distribution $G$ reduces the randomness and makes the problem of identifying the persuasion strategy tractable. 
	
	The identified set of $(\alpha,\beta,G)$ in Definition \ref{def: identified set of the RIDC model} is defined through the rational inattention discrete choice model only, and it ignores the empirical content of the subsequent persuasion stage. There are two reasons to define the identified set in this way. First, if we only have data on markets without any persuader, i.e. $\chi^m=0$ for all markets, the identified set defined in Definition \ref{def: identified set of the RIDC model} can still be used. Second, the unobserved persuasion signal $s^{ID}$ in the persuasion stage makes it hard to characterize the empirical content of the whole persuasion game. I will stick with these two definitions and characterize the corresponding moment conditions.

	\subsection{Moment Equality Model }
	Recall that given $(\alpha,\beta)$ and the unconditional choice probability $\{\mathcal{P}_j^{0,k}(\mathbf{X})\}_{j,k}$, the model predicted market share is given by (\ref{eq: market share is the weighted sum of group share}). Following BLP, I denote  $\delta^m_j= u_1(x^m_j,\beta)+\epsilon^m_j$ and let $\bm{\delta}^m=(\delta^m_1,...,\delta^m_J)$. Then the predicted market share in (\ref{eq: market share is the weighted sum of group share}) can be written as:
	\[
	\begin{split}
	\mathcal{P}_j^{m}(\alpha,\beta,\bm{\epsilon},D^m,\mathbf{X}^m,\{\mathcal{P}_j^{0,k}(\mathbf{X})\}_{j,k})&=\sum_k\frac{ {\mathcal{P}_j^{0,k}(\mathbf{X}^m)e^{\delta^m_j+u_2(x^m_j,\nu_k,\alpha)}}}{\mathcal{P}_J^{0,k}(\mathbf{X}^m)+\sum_{l=1}^{J-1} {\mathcal{P}_l^{0,k}(\mathbf{X}^m)e^{\delta^m_l+u_2(x^m_l,\nu_k,\alpha)}}} d_k^m\\
	&\equiv ms^*_j(\mathbf{X}^m,\bm{\delta}^m,\alpha,D^m,\{\mathcal{P}_j^{0,k}(\mathbf{X})\}_{j,k}),
	\end{split}	
	\]
	where the utility of the outside option is normalized to zero, so $\delta_J^m=u_2(x_J^m,\nu_k,\alpha)=0$. Consider a mapping $T:\mathbb{R}^{J-1}\rightarrow\mathbb{R}^{J-1}$ such that 
	{\footnotesize
		\begin{equation}\label{eq: Contration Mapping}
		\left[T[\mathbf{X}^m,\bm{ms}^m,\alpha,D^m,\{\mathcal{P}_j^{0,k}(\mathbf{X})\}_{j,k}](\bm{\delta}^m)\right]_j=\delta_j+\log(ms^m_j)-\log(ms^*_j(\mathbf{X}^m,\bm{\delta}_j,\alpha,D^m,\{\mathcal{P}_j^{0,k}(\mathbf{X})\}_{j,k})),
		\end{equation}}
	where $[T]_j$ is the $j$-th entry in the output vector. The input $\bm{ms}^m$ will be the observed market share. When $T(\bm{\delta}^m)=\bm{\delta}^m$, the observed market share $\bm{ms}^m$ equals the model predicted market share. This map is a contraction mapping whenever the outside option has a nonzero unconditional choice probability. As a result, there exists a unique market level $\bm{\delta}^m$ that matches the observed market share with the model predicted market share.
	\begin{lem}\label{lem: Contraction Mapping Lemma}
		Suppose in a market we have $ms_J^0>0$, then the mapping defined by (\ref{eq: Contration Mapping}) is a contraction mapping. Let $\bm{\delta}^{*m}$ denote the fixed point of the contraction mapping (\ref{eq: Contration Mapping}). As a result, the unobserved heterogeneity $\bm{\delta}^{*m}$ is a function of observables $\mathbf{x^m}$, $D_k$, $\bm{ms}^m$ and the parameters $\alpha$ and  $P_j^{0,k}(\mathbf{X})$.
	\end{lem}
	Now I state the first identification result of the prior distribution $G$.
	\begin{prop}\label{prop: identification of G when knowing alpha beta}
		For each $(\alpha,\beta,\{\mathcal{P}_j^{0,k}(\mathbf{X})\}_{j,k})$ in the identified set $\Gamma_I$ defined in Definition \ref{def: identified set of the RIDC model}, there exists a unique $G^*$ such that 	
		$(\alpha,\beta,G^*,\{\mathcal{P}_j^{0,k}(\mathbf{X})\}_{j,k})\in \Gamma_I$. In particular, for any measurable set $B$, define the set
		{ \footnotesize
			\[
			MS(B;\mathbf{x^m}, D_k, \bm{ms}^m,\alpha,P_j^{0,k}(\mathbf{X}),\beta)\equiv\{\bm{ms}^m: \bm{\delta}^{*m}(\mathbf{x^m}, D^m, \bm{ms}^m,\alpha,P_j^{0,k}(\mathbf{X}))- [u_1(x_j^m,\beta)]_{j=1}^J\in B\},
			\]}
		where $\bm{\delta}^{*m}$ is defined in Lemma \ref{lem: Contraction Mapping Lemma} and $[u_1(x_j^m,\beta)]_{j=1}^J=[u_1(x_1^m,\beta),...,u_1(x_J^m,\beta)]'$. The $G^*$ satisfies 
		\begin{equation}\label{eq: Identified G, conditioned on alpha beta}
		Pr_{G^*}(\bm{\epsilon}^m\in B)= Pr_{\mathcal{F}_{\chi=0}}(
		\bm{ms}^{m}\in MS(B;\mathbf{x^m}, D_k, \bm{ms}^m,\alpha,P_j^{0,k}(\mathbf{X}),\beta)).
		\end{equation}
	\end{prop}
	
	\begin{proof}
		I prove this statement by contradiction. Suppose there exists a $G'\ne G^*$ such that $(\alpha,\beta,G',\{\mathcal{P}_j^{0,k}(\mathbf{X})\}_{j,k})$ is also in the identified set. Suppose there exists a positively measured set $B'$ such that 
		\[
		Pr_{G^*}(\bm{\epsilon}^m\in B')\ne Pr_{G'}(\bm{\epsilon}^m\in B').
		\]
		I claim the distribution of $\bm{ms}^m$ implied by $G'$, denoted by $\mathcal{F}'_{\chi=0}$ is different from $\mathcal{F}_{\chi=0}$. By equation (\ref{eq: market share is the weighted sum of group share}),
		\[
		\begin{split}
		&\quad Pr_{\mathcal{F}'_{\chi=0}}(
		\bm{ms}^{m}\in MS(B';\mathbf{x^m}, D_k, \bm{ms}^m,\alpha,P_j^{0,k}(\mathbf{X}),\beta))\\
		&= Pr_{G'}(\bm{\epsilon}^m\in B')\\
		&\ne Pr_{G^*}(\bm{\epsilon}^m\in B')\\
		&= Pr_{\mathcal{F}_{\chi=0}}(
		\bm{ms}^{m}\in MS(B';\mathbf{x^m}, D_k, \bm{ms}^m,\alpha,P_j^{0,k}(\mathbf{X}),\beta)).
		\end{split}
		\]
		Therefore, $G'$ cannot generate the same data distribution $\mathcal{F}_{\chi=0}$, so $G'$ is not in the identified set by Definition \ref{def: identified set of the RIDC model}.
	\end{proof}

	Proposition \ref{prop: identification of G when knowing alpha beta} states that once we know $(\alpha,\beta,\{\mathcal{P}_j^{0,k}(\mathbf{X})\}_{j,k})$,  the distribution $G$ is point identified. This is similar to the identification strategy in the first price auction models \citep{GPV2000}. The $\bm{\delta}^{*m}-[u_1(x_j^m,\beta)]_{j=1}^J$ is the pseudo value of $\bm{\epsilon}^m$, similar to the pseudo value that is constructed from bids in the auction model. 
	
	\begin{prop}\label{prop: moment conditions from RIDC model}
		Suppose assumptions \ref{assum: (DGP without Persuasion) }, \ref{assum: normalization}. Suppose the unconditional choice probability $\{\mathcal{P}_j^{0,k}(\mathbf{X})\}_{j,k}$ are uniformly bounded away from zero and one. Each  $(\alpha,\beta,\{\mathcal{P}_j^{0,k}(\mathbf{X})\}_{j,k},G)$ in the identified set $\Gamma_I$ defined in Definition \ref{def: identified set of the RIDC model} satisfies:
		\begin{enumerate}
			\item Constraint on unconditional choice probability:
			
			\begin{equation} \label{Moment condition: P}
			E\left[\bm{ms}^{m}-\begin{pmatrix}
			\mathcal{P}_1^{0,1}(\mathbf{X}^m)& ... &\mathcal{P}_1^{0,K}(\mathbf{X}^m) \\ 
			...& ... &... \\ 
			\mathcal{P}_J^{0,1}(\mathbf{X}^m)& ... & \mathcal{P}_J^{0,K}(\mathbf{X})^m
			\end{pmatrix}
			\begin{pmatrix} 
			d^m_1\\
			...\\
			d^m_K
			\end{pmatrix} \Bigg|(D^m),\mathbf{X}^m\right]=0;
			\end{equation}
			\item Instrument constraint:
			\begin{equation}\label{moment condition: delta}
			E[\delta^*_j(\bm{ms}^m,\mathbf{X}^m,D^m,\alpha,\{\mathcal{P}_j^{0,k}(\mathbf{X})\}_{j,k})-u_1(X_j^m,\beta)|\mathbf{X}^m,D^m]=0 ,\, \, \, \forall j=1,...,J-1;
			\end{equation}
			\item Optimality constraint on $\{\mathcal{P}_j^{0,k}(\mathbf{X})\}_{j,k}$ $\forall \,\,j=1,2,..J-1\,\,k=1,...,K$:
			\begin{equation} 
			\label{Moment: Fixed point first order condition}
			E \bigg[\frac{ {e^{\delta^m_j+u_2(x_j^m,\nu_k,\alpha)}}}{\sum_{l\in\mathcal{J}} {\mathcal{P}_l^{0,k}(\mathbf{X})e^{\delta^m_l+u_2(x_l^m,\nu_k,\alpha)}}}-1\bigg| \mathbf{X}^m\bigg]=0;
			\end{equation}
			\item $G$ satisfies equation (\ref{eq: Identified G, conditioned on alpha beta}).
			
		\end{enumerate}
	\end{prop}
	
	The first moment equality (\ref{Moment condition: P}) is equivalent to condition (\ref{eq: mean of con is unc}) in Definition \ref{def: identified set of the RIDC model}, since $\bm{ms}^m$ is the conditional choice probability while $\mathcal{P}_j^{0,k}$ is the unconditional choice probability. The second moment inequality (\ref{moment condition: delta}) is the consequence of conditions 2 and 3 in Assumption \ref{assum: (DGP without Persuasion) }. The third moment inequality is the first order condition of (\ref{eq: alterantive optimization}).

	\begin{remark}
		The identification results are different from the results in BLP in several ways. First, we need the number of markets to be large to identify the unconditional choice probability for different demographic groups from (\ref{Moment condition: P}). From the identified unconditional choice probability, we can proceed to identify coefficients on the product and demographic heterogeneous characteristics $\alpha$ and $\beta$. Second, in BLP we assume there is a vector of unobserved product heterogeneity $\bm{\xi}={(\xi_1,...\xi_J)}$ that can be recovered by matching market shares and model prediction. In the rational inattention discrete choice model, we recover a vector of market-specific utility shock $\bm{\epsilon}$. Third, the prior distribution of $\bm{\epsilon}$ is the structural object that we are interested in, but the distribution of $\bm{\xi}$ in BLP is not of fundamental interest. 
	\end{remark}
	
	\begin{remark}
		If the price of item $j$, denoted by $q_j$, enters in the product heterogeneity $X_j$, then the price is likely to be correlated with the unobserved market realized utility shock. For example, when sellers know the realization of $\bm{\epsilon}$, they may set a price accordingly. In this case, the assumption $E[\bm{\epsilon}^m|X^m]=0$ fails. In this case, we may want to find an instrument for $q_j$. The choice of instruments for the price is discussed in BLP. 
	\end{remark}

	Definition \ref{def: identified set of the persuasion strategy} of the identified set of persuasion strategy is conditioned on the value of $(\beta,\alpha,G)$. If $(\beta,\alpha,G)$ is point identified from Proposition \ref{prop: moment conditions from RIDC model}, we can assume that $(\beta,\alpha,G)$ is known by the econometrician and plug the identified $(\beta,\alpha,G)$ into Definition \ref{def: identified set of the persuasion strategy}. If $(\beta,\alpha,G)$ is not point identified, we can do analysis by considering that each point in the identified set $\Gamma_I$ as the true value separately. 
	
	For a point $(\alpha,\beta,G)$ in the identified set $\Gamma_I$, equation (\ref{eq: condi choice prob with persuasion}) defines the conditional choice probability of demographic group $k$ choosing item $j$ when they receive a persuasion signal $s$ from the ID. The $\tilde{\mathcal{P}}_{j,s}^{0,k}$ is the unconditional choice probability solved from (\ref{eq: RI optimization})-(\ref{mutual information}) when the intermediate belief is $\tilde{F}(\bm{\epsilon}|s)$. The choice probability $\tilde{\mathcal{P}}_{j,s}^{0,k}$ is conditioned on the signal $s^{ID}$, but unconditional on the utility shock. 
	
	The observed market share $\tilde{\bm{ms}}^m$ is a linear combination of different demographic groups' conditional choice probability :
	\begin{equation}\label{eq: observed choice prob in markets with persuasion}
	\tilde{{ms}}^m_j=(\tilde{\mathcal{P}}_{j,s}^{1}(\bm{\epsilon},\mathbf{X}^m),...\tilde{\mathcal{P}}_{j,s}^{K}(\bm{\epsilon},\mathbf{X}^m))
	(d^m_1,...,d^m_K)^\prime.
	\end{equation}
	Conditioned on $(d_1^m,...,d_K^m)$, we can take expectation on both sides of (\ref{eq: observed choice prob in markets with persuasion}) to get:
	\begin{equation} \label{market equality for persuasion}
	E[\tilde{ms}_j-(\tilde{\mathcal{P}}_{j,s}^{1}(\bm{\epsilon},\mathbf{X}^m),...\tilde{\mathcal{P}}_{j,s}^{K}(\bm{\epsilon},\mathbf{X}^m))(d^m_1,...,d^m_K)^\prime|D^m,\mathbf{X}^m] =0,\quad \forall j=1,...J.
	\end{equation}
	Since we do not observe the realization of the persuasion signal and the realization of the utility shock in each market, we can integrate it out. Let
	\begin{equation}\label{eq: definition of h_j^k}
	\begin{split}
	h_j^k(\mathbf{X}^m;\tilde{F}^k)&:= \int_{(s,\bm{\epsilon})} \tilde{\mathcal{P}}_{j,s}^{k}(\bm{\epsilon},\mathbf{X}^m) d\tilde{F}(s,\bm{\epsilon})\\
	&= \int_{(s,\bm{\epsilon})} \tilde{\mathcal{P}}_{j,s}^{k}(\bm{\epsilon},\mathbf{X}^m) d\tilde{F}(\bm{\epsilon}|s)d\tilde{F}^k(s)\\
	&=\int_s \tilde{\mathcal{P}}_{j,s}^{0,k}(\mathbf{X}^m) d\tilde{F}^k(s)
	\end{split}
	\end{equation}
	be the unconditional choice probability for demographic group $k$ under persuasion strategy $\tilde{F}(s,\bm{\epsilon};\theta)$. The third equality holds because $G(\bm{\epsilon}|s;\theta)=\tilde{F}(\bm{\epsilon}|s;\theta)$ by Bayes' rule. 
	\begin{prop}
		Under assumption \ref{assum: (DGP without Persuasion) } - \ref{assum: DGP with Persuasion}, for each $(\alpha,\beta,G)$, the true persuasion strategy parameter $\theta_0$ must satisfy the moment condition 
		
		\begin{equation}\label{eq: moment condition for persuasion}
		E[\tilde{ms}_j-\sum_{k=1}^K h_j^k(\mathbf{X}^m;\tilde{F})d^m_k|D^m,\mathbf{X}^m]=0 \,\,\, \forall j=1...J-1.
		\end{equation}
	\end{prop}
	
	\begin{proof}
		By assumption \ref{assum: (DGP without Persuasion) }, the independence of demographic distribution $D^m$ and $(\bm{\epsilon}^m,\mathbf{X}^m)$: \[E[\tilde{\mathcal{P}}_{j,s}^k(\bm{\epsilon}^m,\mathbf{X}^m)|D^m,\mathbf{X}^m]=\tilde{\mathcal{P}}^{0,k}_{j,s}(\mathbf{X}^m).\] 
		Then by (\ref{market equality for persuasion}), we have 
		\begin{equation} 
		E[\tilde{ms}_j-(\tilde{\mathcal{P}}_{j,s}^{0,1}(\mathbf{X}^m),...\tilde{\mathcal{P}}_{j,s}^{0,K}(\mathbf{X}^m))(d^m_1,...,d^m_K)^\prime|D^m,\mathbf{X}^m] =0.
		\end{equation}
		Since the signal $s^{ID}\perp{D^m,\bm{X}^m}$ by assumption \ref{assum: DGP with Persuasion}, we have $E[\tilde{\mathcal{P}}^{0,k}_{j,s}(\mathbf{X}^m)|D^m,\mathbf{X}^m]=h_j^k(\mathbf{X}^m;\tilde{F})$. The result follows.
	\end{proof}
	
	The effective number of conditional moment equality is $J-1$ since I have the constraint that $\sum \tilde{ms}_j=1$. 
	We should be careful with the persuasion strategy in Bayesian persuasion.  The value of a signal in persuasion strategy itself has no meaning beyond the context of a communication game. For example, if $\tilde{F}_1$ is the distribution of $(\bm{\epsilon},s^{ID})$ and is the persuasion strategy used by the persuader, then let $\tilde{F}_2$ be the distribution of $(\bm{\epsilon},s^{ID}+\Delta)$, where $\Delta$ is an arbitrary vector that lies in the same space as $s^{ID}$. $\tilde{F}_2$ as a persuasion strategy is not different from $\tilde{F}_1$ since the value of the signal does not matter.

	
	In practice, we can consider the case where the persuasion strategy is indexed by a finite-dimensional parameter $\theta$: $\tilde{F}^k(s^{ID},\bm{\epsilon};\theta)$, and the support of $s^{ID}$ is finite. The persuasion strategy can depend on the demographic group $k$. There are several justifications for the use of a parametric persuasion strategy. First, when there are only two choices, the optimal persuasion strategy is to use a cut-off rule, see \citet{Kamenica2016}. In this case, the parameter $\theta$ is the cutoff points, and signals only take two values. Second, in many empirical contexts, it is costly to design complex persuasion strategies. For example, an online advertisement can only send a simple signal within a few seconds. If the cost of signal increase with the number of parameters and support points of signals, it is natural to restrict the persuasion strategy to parametric form. Third, a parametric persuasion strategy with discrete signal support facilitates a clear interpretation of the meaning of the signals. In \citet{Kamenica2011}, signals are interpreted as action recommendations.

	\subsubsection*{Discussion of Moment Condition (\ref{eq: moment condition for persuasion})}
	One issue with the moment condition (\ref{eq: moment condition for persuasion}) is that it does not guarantee the identification of persuasion  parameters $\theta$. For example, consider the case where there is only one  demographic group $K=1$ and no product characteristics heterogeneity across markets $\mathbf{X}^m=\mathbf{X}\,\,\forall m$. In this case, moment condition (\ref{eq: moment condition for persuasion}) implies $h_j(\tilde{F})=E[\tilde{ms}_j]$. If $\tilde{F}$ is indexed by a parameter $\theta$  and $h_j(\tilde{F}(\theta))$ is not monotone in $\theta$, then $\theta$ is not necessarily point identified. 
	
	There are several restrictions that help to tighten the identified set of $\tilde{F}$. The first is to impose the persuasion strategy is the same for certain demographic groups, i.e. $\tilde{F}^k(s|\bm{\epsilon})=\tilde{F}^{k'}(s|\bm{\epsilon})$ for some $k\ne k'$. Then demographic variation will tighten the bounds on the persuasion strategy. This is because different demographic groups' choice probability can have different sensitivity to the same persuasion strategy. The second is to impose the parameter $\theta$ to be of lower dimension smaller than $J$. The variation of the choice probability across different products can tighten the bounds of the parameter that indexes the persuasion strategy. Third, the variation of product characteristics across markets can also tighten the bounds on $\tilde{F}$. This is because if in a market $m$ the $j$-th product characteristics $x_j^m$ generates large utility to decision makers, persuasion strategy is unlikely to change the market share a lot. 
	
	\subsubsection*{Point Identification Assumption}
	
	It is worthwhile to discuss the assumptions under which parameters $(\alpha,\beta,\mathcal{P}_j^{0,k},G)$ and $\theta$ are point identified. Note that the moment conditions constructed in (\ref{Moment condition: P})- (\ref{Moment: Fixed point first order condition}) are similar to the moment conditions appeared in BLP, except that I have extra parameters $\mathcal{P}_{j}^{0,k}(\mathbf{X})$ to identify. Note that the moment condition for $\mathcal{P}_j^{0,k}(\mathbf{X})$ is similar to the moment condition for linear regression, so if $E[F_k F_k^\prime|\mathbf{X}]$ is invertible $\mathbf{X}-a.s.$, then $\mathcal{P}_{j}^{0,k}(\mathbf{X})$ is identified. The global sufficient primitive conditions for identification of moment conditions (\ref{moment condition: delta})- (\ref{Moment: Fixed point first order condition}) are not easy to interpret, because the fixed point $\bm{\delta}^*$ in Lemma \ref{lem: Contraction Mapping Lemma} is highly non-linear in its arguments. In a similar situation in BLP, they assume the moment conditions are sufficient to identify the utility parameter. 
	\begin{assumption}\label{assum: point id of alpha beta}
		(Identification Assumption)	
		
		\begin{enumerate}
			\item  $E[F_k F_k^\prime|\mathbf{X}]$ is invertible, $\mathbf{X}-a.s$ .
			\item At the true parameter  $\{\mathcal{P}_j^{0,k}(\mathbf{X})\}_{j,k}$, there is a unique $(\alpha,\beta)$ such that moment conditions (\ref{moment condition: delta}) and (\ref{Moment: Fixed point first order condition}) hold.
		\end{enumerate}
	\end{assumption}
	The second requirement in Assumption \ref{assum: point id of alpha beta} is not as restrictive as it seems. In particular, if there is only one demographic group, the fixed point in Lemma \ref{lem: Contraction Mapping Lemma} is given by 
	\begin{equation}\label{eq: closed for solution of delta when no demographic}
	\delta_j^*= \log\frac{ms_j^m}{ms_J^m}- \log\frac{\mathcal{P}_j^{0}(\mathbf{X}^m)}{\mathcal{P}_J^{0}(\mathbf{X}^m)},
	\end{equation}
	and moment condition (\ref{moment condition: delta}) becomes 
	\[
	E\left[\log\frac{ms_j^m}{ms_J^m}- \log\frac{\mathcal{P}_j^{0}(\mathbf{X}^m)}{\mathcal{P}_J^{0}(\mathbf{X}^m)}-u_1(X_j^m,\beta)\bigg|\mathbf{X}^m\right]=0.
	\]
	If $\{\mathcal{P}_j^{0,k}(\mathbf{X})\}_{j,k}$ is identified from moment condition \ref{Moment condition: P} and $u_1$ is a linear function, then $\beta$ is point identified.

	Now suppose the persuasion is parametric and indexed by $\theta$. The assumptions to guarantee that $\theta$ is identified up to $(\alpha,\beta,G)$\footnote{We say $\theta$ is identified up to $(\alpha,\beta,G)$ if the data generating process allow us to point identify $\theta$ for each given parameter $(\alpha,\beta,G)$. } is easier to write down. The discussion of (\ref{eq: moment condition for persuasion}) shows that $h_j^{k,0}(x^m)\equiv h_j^k(x^m;\theta_0)$ is identified if $E[D^m D^{m\prime}|\mathbf{X}]$ is invertible, $\mathbf{X}-a.s$, where $\theta_0$ is the true value of $\theta$. Then the identified set of the persuasion strategy is then the set of the $\theta^*$ such that $h_j^k(\bm{x};\theta^*)=h_j^{k,0}(\bm{x})$ for all $j,k$ and $\bm{x}\in supp(\mathbf{X})$.
	\begin{assumption}
		The matrix $E[D^m D^{m\prime}|\mathbf{X}]$  is invertible  for $\mathbf{X}-a.s.$.
	\end{assumption}

	\section{Estimation}
	Under Assumption \ref{assum: point id of alpha beta}, $(\alpha,\beta,G,\{\mathcal{P}^0_{j,k}(\mathbf{X})\}_{j,k})$ is point identified from moment conditions (\ref{Moment condition: P}), (\ref{moment condition: delta}) and (\ref{Moment: Fixed point first order condition}). When the product characteristics $\mathbf{X}$ are continuously distributed,  $\mathcal{P}_{j}^{0,k}(\mathbf{X})$ in moment condition  (\ref{Moment condition: P}) needs to be estimated non-parametrically. However, in some empirical settings, the product characteristics are discrete and standard estimators of moment equality such as GMM estimator can be implemented directly. In this section, I discuss the estimation of $(\alpha,\beta,G,\{\mathcal{P}^0_{j,k}(\mathbf{X})\}_{j,k})$ when the characteristics $\mathbf{X}$ are discrete. 
	\begin{assumption}\label{assum: discrete X}
		The product characteristics $\mathbf{X}^m$ are discretely distributed and supported on $L$ points: $\{\bm{x}(1),...\bm{x}(L)\}$, and the probability $\inf_{l=1,...,L}Pr(\mathbf{X}^m=x(l))>1/C$ for some constant $C>0$.
	\end{assumption}
	
	Under Assumption \ref{assum: discrete X}, the analysis of moment conditions  (\ref{Moment condition: P}), (\ref{moment condition: delta}) and (\ref{Moment: Fixed point first order condition}) can be done conditioned on the value of $\mathbf{X}^m$ separately. Since the demographic characteristics $v_k$ are also discrete, the most general utility function  of (\ref{eq: parameteric utility assumption}) under discrete $v_k$ and $\mathbf{X}$ can be re-written as 
	\[
	u_{ijkm}= \alpha_j^k(l)\quad if \quad (x_j^m)_{j=1}^J=\bm{x}(l),
	\]
	where $\alpha_j^k(l)$ is the mean utility\footnote{The utility parameter $\beta$ cannot be separated from $\alpha_j^k(l)$, so I normalize $u_1\equiv0$ for all $j$.} of product $j$ for a demographic group $k$ individual in a market with characteristics $\bm{x}(l)$. Any parametric assumption on the utility $u_1$ and $u_2$ in (\ref{eq: parameteric utility assumption}) can be imposed as constraints on the value of $\alpha_j^k(l)$.
	
	Even if $\nu_k$ is distributed on $K$ discrete points, the random vector $D^m$ is continuously distributed. Moment conditions (\ref{Moment condition: P}), (\ref{moment condition: delta}) are still conditioned on $D^m$ and we need to transform them into unconditional moment conditions. Moment condition (\ref{Moment condition: P}) is linear in the elements of $D^m$ and the optimal instrument will be $d_1^m,...,d_k^m$, and we can define  $\mathcal{P}_{j}^{0,k}(\bm{x}(l))$ as $\mathcal{P}_{j}^{0,k}(l)$. For moment condition (\ref{moment condition: delta}), we can use $D^m$ and its second order power terms $\{(d_j^m)^t: j=1,...J, t=1,2 \}$ as instruments to form unconditional moment conditions. 
	
	Let $\bm{\alpha}$ denote the vector of $\{\alpha_j^k(l)\}_{j,k,l}$ and $\mathbf{P}$ denote the vector of $\{\mathcal{P}_{j}^{0,k}(l)\}_{j,k,l}$. Let $\gamma(\bm{ms}^m,D^m,\mathbf{X}^m,\bm{\alpha},\mathbf{P})$ denote the moment unconditional conditions. The standard GMM estimator of $(\bm{\alpha},\mathbf{P})$ is given by
	\begin{equation}\label{eq: GMM ojbective for alpha and P}
	(\hat{\bm{\alpha}},\hat{\mathbf{P}})=\arg\min [\frac{1}{M}\sum_{m=1}^M\gamma(\bm{ms}^m,D^m,\mathbf{X}^m,\bm{\alpha},\mathbf{P})]^\prime \hat{W} [\frac{1}{M}\sum_{m=1}^M\gamma(\bm{ms}^m,D^m,\mathbf{X}^m,\bm{\alpha},\mathbf{P})],
	\end{equation}
	where $M$ is the number of markets without the persuader, and $\hat{W}$ is any positive semi-definite weighting matrix. Standard asymptotic normality results\footnote{For example, see Theorem 3.4 of \citet{newey1994large}.} on the GMM estimator can be applied if the moment condition satisfies some regularity conditions.
	\begin{assumption} \label{assum: High Level consistency}
		Suppose the following conditions hold: (i).The true parameter value $(\bm{\alpha}^0,\mathbf{P}^0)$ lies in the interior of the parameter space; (ii). $\gamma(\bm{ms}^m,D^m,\mathbf{X}^m,\cdot,\cdot)$ is continuously differentiable on the interior of the parameter space for $(\bm{ms}^m,D^m,\mathbf{X}^m)$; (iii). $\gamma(\bm{ms}^m,D^m,\mathbf{X}^m,\bm{\alpha}^0,\mathbf{P}^0)$ has finite second moment; (iv) $E[|\nabla_{(\bm{\alpha},\mathbf{P})} \gamma(\bm{ms}^m,D^m,\mathbf{X}^m,\bm{\alpha}^0,\mathbf{P})^0]$ has rank $dim((\bm{\alpha},\mathbf{P}))$; (v). There exists a integrable function $b$ such that  \[\left|\nabla_{(\bm{\alpha},\mathbf{P})} \gamma(\bm{ms}^m,D^m,\mathbf{X}^m,\bm{\alpha},\mathbf{P})\right|<b(\bm{ms}^m,D^m,\mathbf{X}^m).\]  
	\end{assumption}
	
	Conditions (i), (iii) and (iv) are assumptions on the true value of the parameter of interest $(\bm{\alpha}_0,\mathbf{P}_0)$, which are not verifiable without observing the data distribution. Conditions (ii) and (v) are assumptions on the derivatives of the moment conditions. It is difficult to verify (ii) and (v) because $\bm{\delta}^{m*}$ as a function of $\bm{\alpha}$ and $\mathbf{P}$ is defined through the contraction mapping (\ref{eq: Contration Mapping}). General primitive conditions on the rational inattention model to guarantee that $\bm{\delta}^{m*}$ is continuously differentiable in $\bm{\alpha},\mathbf{P}$ are hard to find. However, when there is no demographic heterogeneity, the $\bm{\delta}^{m*}$ in Lemma \ref{lem: Contraction Mapping Lemma} has a closed from solution (\ref{eq: closed for solution of delta when no demographic}). In this case, the moment conditions (\ref{moment condition: delta}) and  (\ref{Moment: Fixed point first order condition}) can be rewritten as 
	\begin{equation*}
	E\left[\left(\log\frac{ms_j^m}{ms_J^m}- \log\frac{\mathcal{P}_j^{0}(\bm{x}(l))}{\mathcal{P}_J^{0}(\bm{x}(l))}-\alpha_j(l)\right)\mathbbm{1}(\mathbf{X}^m=\bm{x}(l))\right]=0,
	\end{equation*}
	\begin{equation*} 
	E \bigg[\left(\frac{ \frac{ms_j^m}{ms_J^m} / \frac{\mathcal{P}_j^{0}(\bm{x}(l))}{\mathcal{P}_J^{0}(\bm{x}(l))}
	}{\sum_{l\in\mathcal{J}} {\mathcal{P}_l^{0,k}(\mathbf{X})\left[\frac{ms_j^m}{ms_J^m} / \frac{\mathcal{P}_j^{0}(\bm{x}(l))}{\mathcal{P}_J^{0}(\bm{x}(l))}\right]}}-1\right)\mathbbm{1}(\mathbf{X}^m=\bm{x}(l))\bigg]=0.
	\end{equation*}
	If there exists a constant $C>0$ such that  $\mathcal{P}_j^0(\bm{x}(l))>1/C$ holds for all $j,l$, then conditions (ii) and (iv) holds. 
	\begin{lem}\label{lemma: first stage normality}
		Suppose assumption \ref{assum: High Level consistency} holds. Denote $B_0=E[\nabla_{\bm{\alpha},\mathbf{P}} \gamma(\bm{ms}^m,D^m,\mathbf{X}^m,\bm{\alpha},\mathbf{P})]$. Then
		\[\sqrt{M}[(\hat{\bm{\alpha}},
		\hat{\mathbf{P}})-({\bm{\alpha}_0},\mathbf{P}_0)\rightarrow_d N(0,\Sigma),\]
		where $\Sigma=(B_0^\prime WB_0)^{-1}B_0^\prime W\Lambda W B_0^\prime (B_0^\prime WB_0)^{-1}$, and 
		\[
		\Lambda=E[\gamma(\bm{ms}^m,D^m,\mathbf{X}^m,\bm{\alpha},\mathbf{P}) \gamma(\bm{ms}^m,D^m,\mathbf{X}^m,\bm{\alpha},\mathbf{P})^\prime].
		\]
	\end{lem}

	Recall that the moment condition for persuasion strategy in (\ref{eq: moment condition for persuasion}) is derived for each identified value of $(\alpha,\beta,G)$. Now I give an estimator of the persuasion strategy when the estimated $(\hat{\bm{\alpha}},\hat{\mathbf{P}})$ in Lemma \ref{lemma: first stage normality} are directly plugged into (\ref{eq: moment condition for persuasion}). This is a two-step estimation procedure and will not be efficient. I will discuss the complexity of the joint estimation of moment conditions \ref{Moment condition: P}-\ref{Moment: Fixed point first order condition} and (\ref{eq: moment condition for persuasion}) after the plug-in estimator of the persuasion strategy is introduced. 
	
	Given the estimated $(\hat{\bm{\alpha}},\hat{\mathbf{P}})$, we can construct a sample of estimated realized utility 
	\begin{equation}\label{eq: estimated sample of v_jk^m}
	\hat{v}_{j,k}^m(\bm{x}(l))=\sum_{l=1}^L \left[\delta_j(\bm{ms}^{m},\mathbf{X}^m,D^m,\hat{\bm{\alpha}},\hat{\mathbf{P}}^{0})+\hat{\alpha}_j^k(\bm{X}^m)\right] \mathbbm{1}(\mathbf{X}^m=\bm{x}(l))
	\end{equation}
	corresponding to  (\ref{eq: the v_j with linear specification}) and a sample of utility shock
	\begin{equation}\label{eq: estimated shocks}
	\hat{\epsilon}^m_j= \delta_j(\bm{ms}^{m},\mathbf{X}^m,D^m,\hat{\bm{\alpha}},\hat{\mathbf{P}}^{0}).
	\end{equation}
	Fixing the demographic group $k$ and the characteristics $\bm{x}(l)$, the distribution of $\hat{v}_{j,k}^m(\bm{x}(l))$ conditioned on $k$ and $\bm{x}(l)$ is an estimated distribution of realized utility. 
	
	To form moment condition (\ref{eq: moment condition for persuasion}), we first need the unconditional choice probability $\tilde{P}_{j,s}^{0,k}(\mathbf{X}^m)$ in (\ref{eq: definition of h_j^k}) for each demographic group $k$ and for each product characteristics. To get an estimator of $\tilde{P}_{j,s}^{0,k}(\mathbf{X}^m)$, denoted by $\hat{\mathcal{P}}_{j,s}^{0,k}(\mathbf{X}^m)$, we need to solve optimization problem (\ref{eq: alterantive optimization}) with an estimated prior belief. I look at the empirical counterparts of optimization problem (\ref{eq: alterantive optimization}) under persuasion strategy $\tilde{F}(s^{ID},\bm{\epsilon};\theta)$ conditioned on markets with  $\mathbf{X}^m=\bm{x}(l)$:
	\begin{equation} \label{eq: empirical alternative optimization}
	\begin{split}
	\frac{1}{\sum_{m'=1}^{M(\bm{x}(l))} \tilde{F}^k(s^{ID}=s|\hat{\epsilon}^{m'}_j;\theta)}
	&\max_{\{\tilde{\mathcal{P}}^{0,k}_{j,s}\}_{j=1}^J} \sum_{m=1}^{M(\bm{x}(l))}  \log(\sum_{j=1}^J \tilde{\mathcal{P}}^{0,k}_{j,s}(\bm{x}(l)) e^{\hat{v}^m_j}) \times \tilde{F}^k(s^{ID}=s|\hat{\epsilon}^m_j;\theta)\\
	&s.t. \, \, \forall j: \tilde{\mathcal{P}}^{0,k}_{j,s}(\bm{x}(l))\ge 0,	\\
	&\sum_{j=1}^J \tilde{\mathcal{P}}^{0,k}_{j,s}(\bm{x}(l)) =1,
	\end{split} 
	\end{equation}
	where $M(\bm x(l))$ is the number of markets such that $\mathbf{X}^m=x(l)$.
	I implicitly imposed that the marginal distribution of $\tilde{F}^k({\epsilon}_j^m)$ is the empirical distribution of $\hat{\bm{\epsilon}}^m$, and by Bayes' rule  $\frac{\tilde{F}^k(s^{ID}|\hat{\epsilon}^m_j;\theta)}{\sum_{m'=1}^M \tilde{F}^k(s^{ID}|\hat{\epsilon}^{m'}_j;\theta)}$ is the posterior belief when the DM receive a signal $s$. Let $\hat{\mathcal{P}}_{j,s}^{0,k}(\bm{x}(l))$ be the solution to (\ref{eq: empirical alternative optimization}), and denoted the vector $(\hat{\mathcal{P}}_{j,s}^{0,k}(\bm{x}(l)))_{j,k,s,l}$ as $\hat{\mathbf{P}}_s$.
	
	After solving $\hat{\mathcal{P}}_{j,s}^{0,k}(\bm{x}(l))$, we can now write the empirical version of moment condition (\ref{eq: moment condition for persuasion}). Let $N$ be the number of markets with persuasion. Denote $\forall l=0,...L$ and $\forall j=1,...J-1$:
	\begin{equation}\label{eq: moment function of theta with plug in estimator}
	g_{l,j,k}(\theta,\tilde{\mathbf{ms}}^m,D^m,\mathbf{X}^m,\hat{\mathbf{P}}_s)= [\tilde{ms}_j^m-\sum_{d=1}^K h_j^d(\theta,\hat{\mathbf{P}}_s,\bm{x}(l))d_k^m] d_k^m \mathbbm{1}(\mathbf{X}^m=\bm{x}(l)), 
	\end{equation}
	\begin{equation} \label{equation: h-estimator}
	h_j^k(\theta,\hat{\mathbf{P}}_s,\bm{x}(l))=\sum_{s} \bigg[\hat{\mathcal{P}}_{j,s}^{0,k}(\bm{x}(l),\theta) \sum_{m=1}^{N(\bm{x}(l))}\frac{\tilde{{F}}(s|\bm{\epsilon}^m;\theta)}{N(\bm{x}(l))} \bigg]
	\end{equation}
	where $\tilde{\mathbf{ms}}^m$ is a vector of share observation in market $m$, and $N(\bm{x}(l))$ is the number of markets with persuasion such that $\mathbf{X}^m=\bm{x}(l)$. Then we can estimate $\theta$ by the usual GMM estimator:
	\begin{equation} \label{eq: GMM objectives theta}
	\hat{\theta}=\arg\min (\frac{1}{N}\sum_{m=1}^N \mathbf{g}^m(\theta))'{W_2}  (\frac{1}{N}\sum_{m=1}^N \mathbf{g}^m(\theta)),
	\end{equation}
	where $\mathbf{g}^m(\theta))$ is the vector of moment functions $(g_{l,j})_{l,j}$ in (\ref{eq: moment function of theta with plug in estimator}).
	
	In what follows, I derive the consistency of $\hat{\theta}$ when the persuasion strategy  has a smooth parametric form $\tilde{F}(s^{ID,\bm{\epsilon}};\theta)$ and the signal $s^{ID}$ is discrete.
	\begin{assumption} \label{assumption: smooth parametric persuasion strategy}
		The persuasion strategy satisfies that  $\exists\,\,C>0$ for all $s$ value:
		\begin{enumerate}
			\item  $\tilde{F}(s|\bm{\epsilon};\theta)$ is differentiable with respect to $\bm{\epsilon}$, and the gradient is uniformly bounded in $\theta$:
			\[\sup_{\theta\in\Theta,s}\left|\frac{\partial \tilde{F}(s|\bm{\epsilon};\theta)}{\partial \epsilon_j}\right|<C;\]
			\item The $\bm{\delta}^{*m}({\bm{ms}}^m,D^m,\mathbf{X}^m;\bm{\alpha},(\mathcal{P}_j^{0,k}(\mathbf{X}^m))_{j,k})$ defined in Lemma \ref{lem: Contraction Mapping Lemma} satisfies 
			\[
			|\frac{\partial {\delta}^{*m}_j }{\partial \kappa}|<C\quad \forall \kappa\in \{\alpha_{j}^k(l), (\mathcal{P}_j^{0,k}(\bm{x}(l)): j,k,l\}
			\]
			for all values of $\bm{ms}^m,D^m,\mathbf{X}^m$. 
			\item The partial derivatives with respect to the elements of $\theta$ satisfy 	\[\sup_{\bm{\epsilon},s,i}\left|\frac{\partial \tilde{F}(s|\bm{\epsilon};\theta)}{\partial \theta_i}\right|<C.\]

		\end{enumerate}
	\end{assumption}
	The following condition imposes that the instruments $Z(D^m)$ point identify the parameter $\theta$. The point identification conditions are discussed in section 3. 
	\begin{assumption}\label{assumption: GMM identification with theta}
		Let $\mathbf{g}(\theta)= \left(g_{l,j}(\theta,\tilde{\mathbf{ms}}^m,D^m,\mathbf{X}^m,\hat{\mathbf{P}}_s)\right)_{j=1,...,J-1}^{l=1,...,L}$, and define $L(\theta)=g(\theta)'W_2g(\theta)$. The following identification condition hold for all $\zeta>0$
		\[
		\sup_{d(\theta,\theta_0)>\zeta} L(\theta)-L(\theta_0) >0.
		\]
	\end{assumption}

	\begin{prop}\label{prop: Consistency of theta}
		Under assumptions \ref{assum: High Level consistency} -\ref{assumption: GMM identification with theta}, and technical assumption \ref{assumption: M-estimation of P}, $\hat{\theta}$ is a consistent estimator of $\theta_0$. 
	\end{prop}

	\begin{remark}
		The asymptotic distribution of $\hat{\theta}$ is not derived in this paper. There are two difficulties in deriving the asymptotic distribution of $\theta$. The unconditional choice probability vector under persuasion $\hat{P}_s$ is estimated using the sample of markets without persuasion. The sampling error of $\hat{P}_s$ comes from two aspects: (i) $\hat{P}_s$ is estimated from the empirical version (\ref{eq: empirical alternative optimization}) of the optimization problem (\ref{eq: alterantive optimization}); (ii) the utility shocks in (\ref{eq: empirical alternative optimization}) are constructed from the estimator $\hat{\bm{\alpha}}$. Another difficulty comes from the fact that $\mathcal{P}_{j,k}^{0,s}(\bm{x}(l))$ can be local to the boundary to the parameter space under the true persuasion strategy, i.e.  $\mathcal{P}_{j,k}^{0,s}(\bm{x}(l))\approx\frac{1}{\sqrt{n}}$ for some $(j,k,l)$. In this case, the sampling distribution of $\hat{\mathcal{P}}_{j,k}^{0,s}(\bm{x}(l))$ is hard to derive and the influence of the sampling error on $\hat{\theta}$ is hard to derive. 
	\end{remark}
	
	\subsubsection*{Joint Estimation and Two Step Estimation}
	In this section, I briefly discuss how to estimate the persuasion strategy parameter $\theta$ and preference parameters $(\bm{\alpha},\{\mathcal{P}^{0}_{j,k}(\bm{x}(l))\},G)$ jointly using moment conditions (\ref{Moment condition: P})-(\ref{Moment: Fixed point first order condition}) and (\ref{eq: moment condition for persuasion}). The objective function of joint GMM estimation is just the simple stack of $\gamma^m$ in (\ref{eq: GMM ojbective for alpha and P}) and $\mathbf{g}^m$ in (\ref{eq: GMM objectives theta}). For each $(\bm{\alpha},\{\mathcal{P}^{0}_{j,k}(\bm{x}(l))\},\theta)$ in the parameter space, we need to find the $\delta^{*m}$ for each market, and construct the pseudo sample of $\{\epsilon^m\}_{m=1}^M$. Given the pseudo sample of $\{\epsilon^m\}_{m=1}^M$, we then solve the optimization problem (\ref{eq: empirical alternative optimization}) to get $h_j^k$ in (\ref{equation: h-estimator}). Given $\bm{\epsilon}^m$ and $h_j^k$, we can evaluate the value of the joint GMM objective function at this $(\bm{\alpha},\{\mathcal{P}^{0}_{j,k}(\bm{x}(l))\},\theta)$. 
	
	The joint GMM estimation procedure introduces two extra computational burden compared with the two step estimation procedure. First, the fixed point  $\delta^{*m}$ needs to calculated at each $(\bm{\alpha},\{\mathcal{P}^{0}_{j,k}(\bm{x}(l))\}_{j,k,l}\theta)$ parameter evaluation in the joint estimation. In contrast, in the two-step estimation, we find the fixed point for each $(\bm{\alpha},\{\mathcal{P}^{0}_{j,k}(\bm{x}(l))\}_{j,k,l})$. If the dimension of $\theta$ is large, the extra parameter $\theta$ can introduce significant computational burden to the joint estimation. Second, the optimization problem (\ref{eq: empirical alternative optimization}) needs to be solved at each $(\bm{\alpha},\{\mathcal{P}^{0}_{j,k}(\bm{x}(l))\}_{j,k,l},\theta)$ in the joint estimation. In contrast, we plug the estimator $(\hat{\bm{\alpha}},\{\hat{\mathcal{P}}^{0}_{j,k}(\bm{x}(l))\}_{j,k,l})$ into (\ref{eq: moment condition for persuasion}), and the optimization problem (\ref{eq: empirical alternative optimization})  only needs to be solved for each $\theta$. Plugging in the estimator $(\hat{\bm{\alpha}},\{\hat{\mathcal{P}}^{0}_{j,k}(\bm{x}(l))\}_{j,k,l})$ reduces the dimension of the parameter space for the second step GMM estimation.
	
	Joint estimation of $(\bm{\alpha},\{\mathcal{P}^{0}_{j,k}(\bm{x}(l))\}_{j,k,l},\theta)$ also makes the inference of $(\bm{\alpha},\{\mathcal{P}^{0}_{j,k}(\bm{x}(l))\}_{j,k,l})$ difficult. The discussion under Proposition \ref{prop: Consistency of theta} reveals the difficulty of deriving the asymptotic distribution of $\hat{\theta}$. The difficulty comes from the unknown limit distribution of $\hat{\mathcal{P}}_{j,k}^{0,s}(\bm{x}(l))$ when ${\mathcal{P}}_{j,k}^{0,s}(\bm{x}(l))$ is local to zero. The same issue will happen to $(\hat{\bm{\alpha}},\{\hat{\mathcal{P}}^{0}_{j,k}(\bm{x}(l))\}_{j,k,l})$ if we estimate all moment conditions jointly.

	\section{Application: Fox News and the 2000 Presidential Election}
	In this section, I apply the rational inattention, discrete choice model, with persuasion to the effect of Fox News on the 2000 presidential election \citep{DellaVigna2007}. 
	
	Fox News started the distribution of its channel in 1996 and
	its twenty-four-hour cable program penetrated about 20\% of the towns in the United States by Nov, 2000. Fox News channel is perceived to provide political views that are right to the mainstream news channel such as ABC and CNN. In the empirical application, I treat the entry of Fox News into the local cable markets as the presence of the persuader. The DMs' prior distribution $G$ is understood as the prior belief on the presidential candidates under mainstream news channels. The goal is to estimate the preference parameters of each demographic group and the persuasion strategy used by Fox News in these markets. The estimated persuasion strategy can reveal the degree of bias in Fox News program. 
	
	\subsection{Data}
	
	The election outcome data are taken from \citet{DellaVigna2007} and the demographic data are a mixture of the original demographic data in \citet{DellaVigna2007} and the 2000 U.S. census data\footnote{ The education level variable in their data set is not correct for some towns. For example, the proportion of residents with no more than high school education and the proportion of residents with more than high school education sum to greater than 1.}. Each observation consists of a vector of presidential election vote results and a vector of demographic statistics that correspond to a town and an indicator for the presence of Fox News. The presidential election vote result includes the total votes cast, the number of votes for the Democratic Party, and the number of votes for the Republican Party. The demographic statistics include the number of people that are above 18 years old, the gender ratio, the ethnic group decomposition (African American, Hispanic, Asian, etc), and the decomposition by education level. The education level statistics are for eligible voters (18+ years old), but ethnic group statistics incorporate both the adults and children. 
	
	The original demographic data in Vigna and Kaplan is flawed. In about 15\% of the towns, the number of votes cast is more than the number of residents above 18 years old. The issue happens when the town name corresponds to multiple administrative levels. For example, there are some names used for two different townships and cities but in different counties, their match tends to get wrong. I re-match the voting data with the 2000 U.S.  census data to deal with this issue but the problem is not solved completely. There are still about 5\% of the towns that have the inconsistency of votes and adults. As mentioned in \citet{DellaVigna2007}, this may be due to flaws in the process of collecting the election data. 
	
	I follow the data selection procedure in \citet{DellaVigna2007} to discard towns: 1. without CNN news channel; 2. the number of precincts in 2000 differs from that in 1996 by 20\%; 3. the total number of votes in 2000 differs from that in 1996 by 100\%; 4. with multiple cabal systems; 5. the number of people with high school and above is more than the number of adults; 6. the number of votes is greater than the number of adults.

	Throughout the application, I assume the choice set includes $\mathcal{J}=3$ options: $\{Rep,Dem,Out\}$. All votes not cast for the two major parties candidates or adults not registered to vote are grouped into the $Out$ option. 
	
	\subsection{Market Assumptions and Justification}
	First I separate markets into two groups: with persuasion and without persuasion. If Fox News is available in the town, I assume the town is under the influence of the persuader. This assumes that the presence of Fox News influences the whole town. Since I only use observation of towns with one unique cable company, if the cable company includes Fox News, everyone in the town should have access to the channel. While some residents may not watch the channel, the contents of the news program can be spread through workplaces and places of entertainment. This also assumes that towns without Fox News cannot be influenced by persuasion. This assumption suits the historical context in 2000 where the fixed broadband subscription in the United States accounts for around 2.5\% of the population, so streaming of Fox News is not accessible to major voters in the towns without Fox News.  
	
	The key assumption on market without persuasion is the assumption \ref{assum: (DGP without Persuasion) }. The $i.i.d$ assumption on $\bm{\epsilon}^m$ assumes that there is no spatial correlation conditioned on the observed characteristics in the town. This variation in $\epsilon_m$ may come from the geographic location differences of towns and the composition of industries in towns. For example, a policy of cleaner fuel may generate different perceptions in the coal mining towns and forest zone. The independence assumption $\epsilon^m \perp D^m$ assumes the composition of demographics does not influence the prior belief.

	For markets with persuasion, assumption \ref{assum: DGP with Persuasion} requires that Fox News use the same persuasion strategy for all towns, regardless of the demographic composition. This assumption is justified because Fox News is a national program, so the perception of persuasion strategy should be similar for all towns\footnote{Note that this is not a restriction on the entry decision. In fact, Fox News can endogenously choose the town they wanted to provide channels but this is out of the scope of this paper. The model aims to estimate the persuasion strategy used by Fox News but does not model its utility to justify the persuasion and the entry. As long as the persuasion strategy is the same for all towns, the identification argument goes through whether the entry was chosen optimally or exogenous.}.  Last, the assumption that the persuader draw persuasion signal $s^{ID}\sim_{i.i.d} \tilde{F}^k$ says that the signals should be independent for all towns. This assumption is hard to justify since Fox News is a national program. However, Fox News reports on different aspects of the candidates (e.g. foreign policy, economic policy), and each town may only focus on one aspect of a candidate, which may result in an $i.i.d$ persuasion signal across towns.
	
	\subsection{The Specification}
	I assume there are no product characteristics across towns. 
	The utility is $u_{kj}^m=\alpha_{j,k} +\epsilon_{j}^m$, where the
	parameters $\alpha_{j,k}$ are the mean utility of candidate $j$ that differ across demographic group $k$. The utility for the outside option is normalized to be zero. I partition the decision makers in each town based on their education level at the time of the election: \{ High School and Lower, College Partial, College Complete\}.\footnote{A finer partition of the demographics is desired, but the U.S. census data do not provide the joint distribution of education with other demographic characteristics.} The segment of education level can reflect the differences in income levels and the political spectrum. The estimators and the 95\% confidence intervals are reported in table \ref{Baseline Model: Estimated Parameters}.

	\begin{table}[h]
		\centering
		\caption{ Estimated Mean Preference Parameters}
		\label{Baseline Model: Estimated Parameters}
		\begin{tabular}{lcccc}
			Choice $j$                                                                              & \multicolumn{3}{c}{$\alpha$}                                                                                                                                                                                       \\ \hline \hline
			& High School                                                          & College Partial                                                      & College Complete                                                     \\
			\multicolumn{1}{c}{Rep}  & \begin{tabular}[c]{@{}c@{}}-0.1318\\ {[}-0.1540,-0.1050{]}\end{tabular} & \begin{tabular}[c]{@{}c@{}}0.1369\\ {[}0.0816,0.1848{]}\end{tabular} & \begin{tabular}[c]{@{}c@{}}0.0306\\ {[}0.0079,0.0538{]}\end{tabular} \\ \hdashline
			\multicolumn{1}{c}{Dem} &  \begin{tabular}[c]{@{}c@{}}-0.0859\\ {[}-0.0983,-0.0707{]}\end{tabular} & \begin{tabular}[c]{@{}c@{}}0.1260\\ {[}0.0693,0.1725{]}\end{tabular} & \begin{tabular}[c]{@{}c@{}}0.0702\\ {[}0.0529,0.0857{]}\end{tabular}
		\end{tabular}   
	\end{table}
	The estimation result shows several interesting observations. First, the group with partial college degree has a slightly lower preference for the Democratic Party than the Republican Party. The partial college group includes eligible voters who earn degrees from community college or technical colleges. So we see that both highly educated group and the least educated group prefer the Democratic Party\footnote{Note that the confidence interval of $\alpha_{Dem,k}$ does not intersect with $\alpha_{Rep,k}$ for $k\in \{High\,\, school,\,\, College\,\,Complete\}$}, but the middle class seems to be indifferent between these two parties. Second, the College Partial group has a higher willingness to vote. However, this does not imply the College Partial group vote more to the Democratic Party than those who complete college education. Table \ref{Unconditional Choice Probability: With and Without Fox News} reports the estimated unconditional choice probability for each demographic group. 
	
	\begin{table}[h]
		\centering
		\caption{Unconditional Choice Probability: With and Without Fox News }
		\label{Unconditional Choice Probability: With and Without Fox News}
		\begin{tabular}{l|cl|lc|cl}
			& \multicolumn{2}{c|}{High School}      & \multicolumn{2}{l|}{College Partial}   & \multicolumn{2}{c}{College Complete} \\ \hline
			& \multicolumn{1}{l}{No Fox} & With Fox & No Fox & \multicolumn{1}{l|}{With Fox} & \multicolumn{1}{l}{No Fox} & With Fox \\
			Rep & 0.1998                     & 0.1610   & 0.5082 & 0.5488                        & 0.3031                     & 0.3415   \\
			Dem & 0.1891                     & 0.2086   & 0.2925 & 0.2498                        & 0.3974                     & 0.3634  
		\end{tabular}
	\end{table}
	The result in table \ref{Unconditional Choice Probability: With and Without Fox News} cannot be generated by a random utility model with Logit shock. By random utility model with Logit shock, we would predict that the College Partial group vote more for the Democratic Party than College Complete group would do, because $\alpha_{Dem,College\,Partial}>\alpha_{Dem,College\,Complete}$. The estimated density of the prior distribution $G$ is given in Figure  \ref{fig:estimated-distribution-of-epsilon}.  
	
	\begin{figure}[H]
		\centering
		\includegraphics[width=0.7\linewidth]{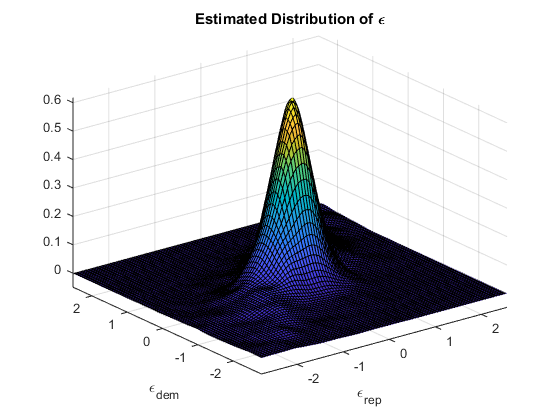}
		\caption{}
		\label{fig:estimated-distribution-of_epsilon}
	\end{figure}

	For the identification of persuasion, I use two different parametric persuasion strategies that differ across the three demographic groups. This can be true if the news programs are designed to target different demographic groups. I restrict the persuader to have only two signals to send. A two-signal persuasion strategy is easy to interpret. A `+' signal means `1 is good' if it only conveys information on ${\epsilon}_1$, and it means `1 is better than 2' if it compares the ${\epsilon}_{1}$ with $\epsilon_{2}$. A `-' signal means the contrary. \footnote{
		Two-signal persuasion strategy is also justified by \citet{Gitmez2018}, where the politician in their model has full control of the news media and voters are heterogeneous in their belief.}
	
	The persuasion strategy of the high school education group is given by 
	\[
	Pr_{\tilde{F}^{HS}}(S^{ID}=-|\bm{\epsilon})=\begin{cases}
	1 \quad &if \quad \epsilon_{rep}< \epsilon_{dem}\\
	\theta_{hs}^{(\epsilon_{rep}-\epsilon_{dem})^2} \quad &if \quad \epsilon_{rep}\ge \epsilon_{dem},
	\end{cases}
	\]
	and the persuasion strategy of the college partial and college complete group is given by
	\[
	Pr_{\tilde{F}^{College}}(S^{ID}=-|\bm{\epsilon})=\begin{cases}
	0 \quad &if \quad \epsilon_{rep}> \epsilon_{dem}\\
	1-\theta_{hs}^{(\epsilon_{rep}-\epsilon_{dem})^2} \quad &if \quad \epsilon_{rep}\le \epsilon_{dem}.
	\end{cases}
	\]

	I use the same parametric family for demographic group with education higher than high school but treat the least educated group separately. This is because table \ref{Unconditional Choice Probability: With and Without Fox News} shows that only the least educated group has decreased unconditional choice probability for the Republican Party and increased unconditional choice probability for the Democratic Party after Fox News entered into their towns.
	
	The `-' signal in the persuasion strategy for the high school group can either mean when the Republican party is indeed worse than the Democratic party, or it can mean with a small probability that the Republican party is better.
	
	The persuasion strategy for the eligible voters with at least a partial college education has a better interpretation. The positive signal $S^{ID}=+$ can be read as `the Republican is better than the Democratic'. A positive signal is always sent when the Republican is indeed better, i.e. $\epsilon_{rep}>\epsilon_{dem}$, and a fake positive signal can also be sent when $\epsilon_{rep}<\epsilon_{dem}$, but the probability decays as the difference becomes larger in absolute value.
	
	The estimated persuasion strategy parameters are reported in table \ref{table: Estimated Persuasion Strategy}, and I plot the probability of the "+" signal for the two persuasion strategies in figure \ref{fig:persuasion-function-college-par-and-com}. We should note that the persuasion strategy parameter $\theta$ is very close to 1 and the entropy of the marginal distribution of the signal is close to zero. The close-to-zero entropy indicates that the signal sent by Fox News does not carry much information. However, the relative scale of entropy is still significantly large compared with the utility parameter $\alpha_{jk}$ for all three groups.

	\begin{figure}[h]
		\centering
		\includegraphics[width=0.9\linewidth]{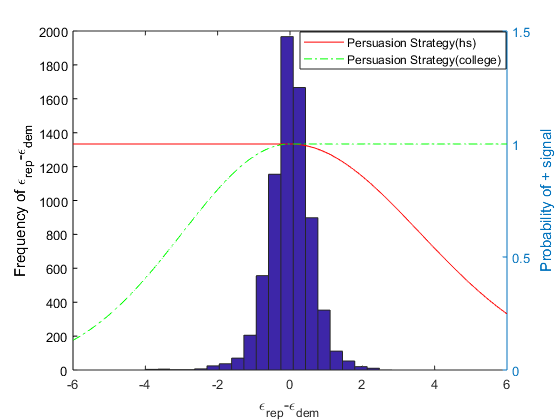}
		\caption{Estimated Probability of Sending "+" Signal and the histogram of $\epsilon_{rep}-\epsilon_{dem}$}
		\label{fig:persuasion-function-college-par-and-com}
	\end{figure}
	
	\begin{table}[H]
		\centering
		\caption{Persuasion Strategy}
		\label{table: Estimated Persuasion Strategy}
		\begin{tabular}{lcc}
			& High School       & College Partial and Complete       \\\hline
			Estimator $\hat{\theta}$ & 0.9620              & 0.9452        \\
			Entropy & 0.0553  &0.0378   
		\end{tabular}
		\begin{tablenotes}
			\small
			\item \textit{Note: The entropy numbers are calculated based on the marginal distribution of signal. }
		\end{tablenotes}	
	\end{table}
	
	The overall fit of the persuasion model can be seen from the difference between the data unconditional choice probability and the unconditional choice probability predicted by the persuasion strategy. Table \ref{table: Model Fit} shows that the model predicts the unconditional choice probability quite well except for the  high school group's unconditional choice probability of choosing the Republican party. 
	\begin{table}[H]
		\centering
		\caption{Unconditional Choice Probability in Towns with Fox News: Model vs Data } 
		\label{table: Model Fit}
		\begin{tabular}{l|cl|lc|cl}
			& \multicolumn{2}{c|}{High School}      & \multicolumn{2}{l|}{College Partial}   & \multicolumn{2}{c}{College Complete} \\ \hline
			& \multicolumn{1}{l} {} Model & Data  & \multicolumn{1}{l}{Model} &Data & \multicolumn{1}{l}{Model} & Data \\
			Rep & \textcolor{red}{0.1853}                    & \textcolor{red}{0.1610}   & 0.5427 & 0.5488                        & 0.3335                     & 0.3415   \\
			Dem & 0.2090                     & 0.2086   & 0.2614 & 0.2498                        & 0.3708                     & 0.3634  
		\end{tabular}
	\end{table}

	\subsection{Welfare Analysis }
	Costly information acquisition can lead the decision maker to choose the second-best choice with some probability. If information is free (i.e.$\lambda=0$, or decision maker can perfectly observe $(\epsilon_{rep},\epsilon_{dem})$), the decision maker should be able to choose the one that maximizes his utility. This is defined as the first-best outcome. Persuasion signal has two influences on decision makers: persuasion signal provides extra information that reduce the entropy of belief, but it also intentionally leads some decision makers to make wrong decisions. In this section, I analyze the welfare by asking what is the percentage of voters that cast votes consistent with their first best choice before and after Fox News enters into their town. Formally, the first best choice $j_k^{m,fb}$ in a town $m$ is defined as 
	\[j_k^{m,fb}=\arg\max_{j\in \mathcal{J}} \alpha_{j,k}+\epsilon^m_j\]  
	and $\mathcal{P}_{j=j^{fb}}^{k}(\bm{\alpha}+\bm{\epsilon}^m)$ is the proportion of voters that make the correct choice in the rational inattention model without persuasion in town $m$, and $\sum_s \tilde{F}(s|\bm{\epsilon}))\mathcal{P}_{j=j^{fb},s}^{k}(\bm{\alpha}+\bm{\epsilon}^m)$ is the proportion of voters that make the correct choice with Fox News Persuasion. Since we have the estimated prior distribution $G(\epsilon_{rep},\epsilon_{dem})$, we get the distribution of $\mathcal{P}_{j=j^{fb}}^k(\bm{\alpha}+\bm{\epsilon}^m)$ and $\sum_s \tilde{F}(s|\bm{\epsilon}))\mathcal{P}_{j=j^{fb},s}^{k}(\bm{\alpha}+\bm{\epsilon}^m)$. The estimated distribution (across towns) can be seen in figure \ref{fig: Welfare dist}. The patterns are quite different for the three groups. For voters with high school education, persuasion does not really help them to make better decisions overall. For voters with a partial college education, persuasion generates higher dispersion in the distribution of voters that vote for their first best choice. It should be noted that even if the persuasion strategy is the same for voters with a partial and full college education, the persuasion strategy tightens the distribution of the first best choice for voters who complete a college education. 
	
	
	\begin{figure}[H]
		\centering
		\caption{ Distribution of percentage of voters that achieve their first best choice}
		\includegraphics[width=1\linewidth]{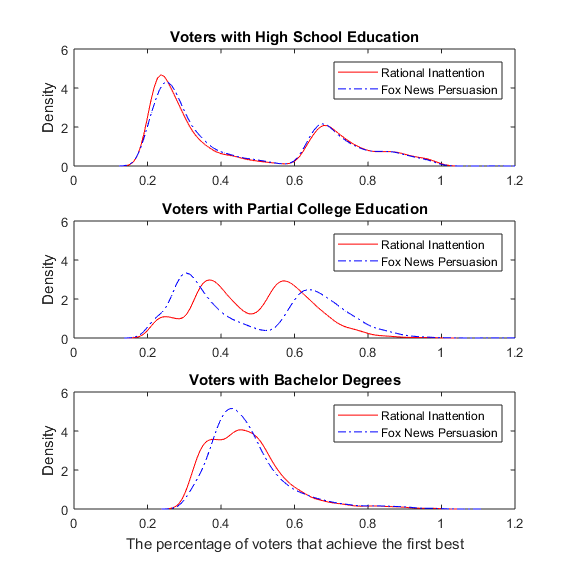}
		
		\label{fig: Welfare dist}
	\end{figure}
	
	%
	%

	\section{Conclusion}
	In this paper, I study the identification of the rational inattention discrete choice model with Bayesian Persuasion. I derive the conditional moment conditions that identify the mean utility of each product and prior distribution. I also show the identification of a parametric persuasion strategy when the persuader plays a sequential game with decision makers in the model. In the empirical application, I studied the effect of Fox News in persuading voters to vote for the Republican Party. I also analyze the welfare change for voters before and after the influence of Fox News.  
	
	For future research, we should derive a method to unify the supply-side model with the identified persuasion strategy. If the supply side, which is Fox News in the context, is rational when it chooses the persuasion strategy, the optimal strategy should reveal constraints on its utility parameters. Such parameters are crucial when we conduct a counterfactual analysis on the supply side. For instance, in the IO context, the preference for persuasion strategy would allow us to model the non-price competition.

	\bibliography{rational_inattention_citation}
	\bibliographystyle{chicago}
	\pagebreak
	
	\pagebreak
	\appendix
	\section{Appendix 1: Data Compression Interpretation of Entropy Cost}\label{section: Data Compression and Entropy Cost}
	
	The entropy of a discrete random variable is closely related to the expected number of binary questions needed to be asked to determine the realization. Consider the following example: 
	\begin{itemize}
		\item $X$ is supported on 4 points: $X_1=(H,H)$, $X_2=(H,L)$, $X_3=(L,H)$ and $X_4=(L,L)$.
		\item The probability of each realization is $P_1=P_4=1/3$ and $P_2=P_3=1/6$.
		\item Consider two ways of asking questions:
		\begin{enumerate}
			\item Q1: The state is: (A) First component is H; (B) the First component is L. Q2: (A) Second component is H; (B) the Second component is L. 
			\item Q1: The state is: (A) Both high; (B) Both low; (C) Neither. Q2: The state is: (A) (H,L); (B) (L,L).		
		\end{enumerate}
	\end{itemize}
	Using the first approach, we need to ask two binary questions for sure to pin down the realization. Using the second approach we are expected to ask one 3-adic question for sure, and with $1/3$ probability we need another binary question. If we consider that a 3-adic question is equivalent to $\log_2 3$ binary questions \footnote{One way to understand this conversion is the following. Suppose we have N binary questions that can cover all possible states of the world, the cardinality of states of the world is approximately $2^N$. In the same case, suppose we need $M$ 3-adic question to cover all states of the world. By setting $2^N\approx 3^M$, we find the $N=M*log_2 3$. The more rigorous conversion argument can be established using large scale data compression theory. See Cover (2006), Chapter 5.}, the expected number of binary question we need to ask is $\log_2 3+ 1/3= -2/3\times\log_2 (1/3)+ 1/3\times\log_2 (1/6) $ which is the entropy number. In many examples, the entropy number cannot be coded with an integer number of binary questions, but nonetheless the entropy number is a good approximation for the complexity of the random variable.
	
	Now, consider the entropy cost function we defined in (\ref{mutual information}). The entropy $H(G)$ is interpreted as the number of binary questions of the prior distribution. Now, given the signal $s$ the DM acquire from the world, the number of binary questions is reduced to $H(F(\cdot|s))$. Since ex ante, the DM do not know the realization of $s$, the expected number of questions remained is $E_s[H(F(\cdot|s))]$. Therefore, the difference of entropy $H(G)-E_s[H(F(\cdot|s))]$ is interpreted as the expected number of binary questions that is answered by the signal $s$, and the unit cost of information $\lambda=0$ is interpreted as the market price for asking a binary question.
	
	The interpretation still works when we consider $s$ being discrete but the state $v$ is continuous. Let's consider an example where $X\sim U[-1,1]$ and $Y=1$ when $X\geq 0$ and $Y=0$ when $X<0$. Given $X$ is negative with probability $0.5$, $Y$ answers one binary question whether $X$ is negative or not. Direct calculation shows that $H(X)=1$ and $H(X|Y)=0$, so the mutual information $I(X;Y)=H(X)-H(X|Y)=1$. 
	
	When the pair $(\mathbf{s},\mathbf{v})$ is continuously distributed, the data compression argument needs to be modified slightly. The approach is to take a quantization of the random variable. The quantization of $\mathbf{v}$ is to slice the support of $\mathbf{v}$ with cubes of side length $\Delta$. As the quantization length $\Delta\rightarrow 0$, the entropy of the discrete random vector, denoted as $V(\delta)$, will converge to the differential entropy of $\mathbf{v}$ in the following sense:
	\[
	H(V(\delta))+\log \Delta^J\rightarrow H(G(\mathbf{v}))\,\,\,\, as\,\,\,\, \Delta\rightarrow 0
	\]
	where $J$ is the dimension of $\mathbf{v}$. We can perform the same quantization for the signal variable $s$. When we calculate the entropy difference $H(G)-E_s[H(F(\cdot|s))]$, which is called the mutual information, the effect of quantization will disappear. See Cover and Thomas(2006), Chapter 8 for discussion of quantization. Then we can use the interpretation for the data compression on the quantized version of $(\mathbf{v},\mathbf{s})$.

	\section{Appendix 2: Proofs of Section 3}
	\subsection{The Contraction Mapping Lemma \ref{lem: Contraction Mapping Lemma}}
	
	\begin{proof}
		The proof is a minor adaption of \citet{Berry1995}. To show the operator $T$ is a contraction mapping, it suffice to show that the conditions of theorem 1 in BLP holds. Let $T_j: R^{J-1}\rightarrow R$ denote the $j-th$ component of the mapping $T:R^{J-1}\rightarrow R^{J-1}$ defined in (\ref{eq: Contration Mapping}). I use the following notation for the proof:
		\begin{equation}\label{eq: append, logit form}
		\mathcal{P}_j^k(\bm{\delta},\mathbf{X},\mathcal{\mathbf{P}}^{0,k},\nu_k,\alpha)=\frac{ {\mathcal{P}_j^{0,k}(\mathbf{X})e^{\delta_j+u_2(X^m,\nu_k,\alpha)}}}{\sum_{l\in\mathcal{J}} {\mathcal{P}_l^{0,k}(\mathbf{X})e^{\delta_l+u_2(X^m,\nu_k,\alpha)}}}.
		\end{equation}
		First note that:
		\begin{align*}
		\frac{\partial T_j}{\partial \delta_j}&=1-\frac{1}{ms^*_j}\times \sum_k \mathcal{P}_j^k(\bm{\delta},\mathbf{X},\mathcal{\mathbf{P}}^{0,k},\nu_k,\alpha)\times(1-\mathcal{P}_j^k(\bm{\delta},\mathbf{X},\mathcal{\mathbf{P}}^{0,k},\nu_k,\alpha))d_k\\
		&\ge 1-\frac{1}{ms^*_j}\times\sum_k \mathcal{P}_j^k(\bm{\delta},\mathbf{X},\mathcal{\mathbf{P}}^{0,k},\nu_k,\alpha)d_k\ge 0\\
		\frac{\partial T_j}{\partial \delta_l}&=\frac{1}{ms^*_j}\times \sum_k \mathcal{P}_j^k(\bm{\delta},\mathbf{X},\mathcal{\mathbf{P}}^{0,k},\nu_k,\alpha)
		\mathcal{P}_l^k(\delta_l,\mathbf{X},\mathcal{\mathbf{P}}^{0,k},\nu_k,\alpha)d_k\ge  0
		\end{align*}
		and for any $j=1,...,J-1$:
		\begin{align*}
		\sum_{l<J} \frac{\partial{T_j}}{\partial \delta_l}&=1-\frac{1}{ms^*_j}\times \sum_k\left[ \mathcal{P}_j^k(\bm{\delta},\mathbf{X},\mathcal{\mathbf{P}}^{0,k},\nu_k,\alpha)\times(1-\sum_{l=1}^{J-1}\mathcal{P}_l^k(\delta_l,\mathbf{X},\mathcal{\mathbf{P}}^{0,k},\nu_k,\alpha))d_k\right]\\
		&= 1-\frac{1}{ms^*_j}\times \sum_k\left[{\mathcal{P}_j^k(\bm{\delta},\mathbf{X},\mathcal{\mathbf{P}}^{0,k},\nu_k,\alpha)\times \mathcal{P}_J^k(\delta_J,\mathbf{X},\mathcal{\mathbf{P}}^{0,k},\nu_k,\alpha)}d_k\right],
		\end{align*}
		where $ms_j^*\equiv ms^*_j(\mathbf{X}^m,\bm{\delta}^m,\alpha,D^m,\{\mathcal{P}_j^{0,k}(\mathbf{X})\}_{j,k})$. For the outside option $J$, $\delta_J=0$ holds because we assume the choice utility of the outside option is zero by normalization. By assumption, the unconditional choice probability of the outside option is non-zero for all $\mathbf{X}$ and all $k$, i.e. $\mathcal{P}_l^{0,k}(\mathbf{X})>0$. The Logit form (\ref{eq: append, logit form}) implies the conditional choice probability $\mathcal{P}_J^k(\delta_J,\mathbf{X},\mathcal{\mathbf{P}}^{0,k},\nu_k,\alpha)>0$ must hold, which further implies:
		\[
		\sum_{l=1}^{J-1} \frac{\partial{T_j}}{\partial \delta_l}<1.
		\]
		This verifies the condition 1 of the contraction mapping theorem in Appendix 1 of BLP.
		
		To verify condition 2 of the contraction mapping theorem in BLP, I rewrite equation (\ref{eq: Contration Mapping}) by plug in the expression of $ms^*_j$ into the mapping $T$:
		\[
		[T(\bm{\delta})]_j=\log(ms_j^m)-\log\left(\sum_k \frac{ {\mathcal{P}_j^{0,k}e^{u_2(X_j^m,\nu_k,\alpha)}}}{ \mathcal{P}_J^{0,k}+\sum_{l=1}^{J-1} {\mathcal{P}_l^{0,k}e^{\delta_l+u_2(x_l,\nu_k,\alpha)}}}d_k\right)
		\]
		The function $T$ is bounded from below by $\log(ms_j^m)-\log\left(\sum_k \frac{ {\mathcal{P}_j^{0,k}e^{u_2(X_j^m,\nu_k,\alpha)}}}{ \mathcal{P}_J^{0,k}}d_k\right)$ when $\delta_j\rightarrow -\infty$.
		
		For condition 3 of the contraction mapping theorem in BLP, for any $j$, I set $\bar{\delta}_j$:
		\[
		\bar{\delta_j}=\arg\min_{\delta_j}\bigg[ ms^m_J-\sum_k\frac{\mathcal{P}_J^{0,k}(\mathbf{X})}{\mathcal{P}_J^{0,k}(\mathbf{X})+\mathcal{P}_j^{0,k}(\mathbf{X})e^{\delta_j+u_2(X_j^m,\nu_k,\alpha)}} d_k \bigg]^2
		\]
		which is the solution of $\delta_j$ to match the market share of the outside option when $\delta_k=-\infty$ for all $k\ne j$. 
	\end{proof}	
	
	\begin{remark}
		The extra condition that the outside option is chosen with positive unconditional choice probability is not required in the proof of \citet{Berry1995}, because when the shock is supported on unbounded space, the outside option will always have a positive choice probability.  The last step is also slightly different from \citet{berry1994} where the unconditional choice probability $\mathcal{P}_j^{0,k}$ appears in the denominator. 
	\end{remark}
	
	\subsection{Proof of Proposition \ref{prop: moment conditions from RIDC model}}
	\begin{proof}
		Since all three moment conditions are conditioned on $\mathbf{X}^m$, and by assumption \ref{assum: (DGP without Persuasion) }, the product characteristics $\mathbf{X}^m$ is independent of the random utility shocks $\bm{\epsilon}^m$ and demographic distribution vector $D^m$, I prove the proposition conditioned on the value of $\mathbf{X}^m$ and drop $\mathbf{X}^m$ moment condition expressions whenever there is no confusion.

		\subsubsection*{Constraint on $\mathcal{P}_{j}^{0,k}$}
		For each market $m$, we observe only the market share vector 
		\[\mathbf{ms}^m=(ms^m_1,...ms^m_J)^{\prime}\] and the demographic distribution 
		\[D^m=(d^{m}_1,...d^{m}_K)\] where $d^m_k$ is the share of people in demographic group $k$ in market $m$. Then in market $m$, the observation $\mathbf{ms}^m$ satisfies:
		\[
		ms^m_j=\sum_{k=1}^K \mathcal{P}_j^k(\bm{\epsilon}^m)d_k^m \quad \forall j=1,...,J.
		\]

		If we take expectation with respect to the $G$ distribution and the demographic distribution on both sides of the above equation, we have
		\begin{equation*}
		E_G[ms^m_j- (\mathcal{P}_j^1(\bm{\epsilon}^m),...\mathcal{P}_j^K(\bm{\epsilon}^m)) (d^m_1,...d^m_K)^\prime| D^m  ]=0
		\end{equation*}
		By assumption \ref{assum: (DGP without Persuasion) }, $(d^m_1,...d^m_K)\perp (\bm{\epsilon}^m,\mathbf{X}^m)$, we have \[E_G[\mathcal{P}_j^k(\bm{\epsilon}^m)d^m_k|(d^m_1,...d^m_K)]=d^m_k E_G[\mathcal{P}_j^k(\bm{\epsilon}^m)]=d^m_k E_G[\mathcal{P}_j^{0,k}].\] 
		Use the linearity of expectation we can rewrite the above equation as:
		\begin{equation*}
		E[ms^m_j- (\mathcal{P}_j^{0,1},...\mathcal{P}_j^{0,K}) (d^m_1,...d^m_K)^\prime |D^m]=0.
		\end{equation*}
		This is the moment condition (\ref{Moment condition: P}). 
		
		\subsubsection*{Independent $\bm{\epsilon}$ constraint}
		Lemma (\ref{eq: Contration Mapping}) establishes $\bm{\delta}^m$ as a function of $(\alpha,\beta,\mathcal{P}_{j}^{0,k})$. So we can write the $\bm{\epsilon}$ as the difference of $\delta$ and $u_1$. 
		The moment condition (\ref{moment condition: delta}) then comes directly from the assumption that $\bm{\epsilon}^m\perp D^m$ in assumption (\ref{assum: (DGP without Persuasion) }).
		
		\subsubsection*{Optimality constraint}
		Lastly, I derive the condition that is implied by the fact that $\mathcal{P}_j^{0,k}$ solves the optimization problem (\ref{eq: alterantive optimization}). Since $\mathcal{P}_j^{0,k}$ uniformly bounded away from zero and one, so the first order condition of (\ref{eq: alterantive optimization}) is
		\begin{equation*}
		\int_{\bm{\epsilon}} \frac{e^{\delta^m_j+u_2(x_j^m,\nu_k,\alpha)}}{\sum_{l=1}^J \mathcal{P}_l^{0,k} e^{\delta^m_l+u_2(x_l^m,\nu_k,\alpha)} } dG(\bm{\epsilon})=1.
		\end{equation*}
		Note that the optimization (\ref{eq: alterantive optimization}) is a convex optimization so the first order condition is sufficient to characterize the solution. So the above first order condition can be transformed into the condition:
		\begin{equation*} 
		E \bigg[\frac{ {e^{\delta_j+u_2(x_j^m,\nu_k,\alpha)}}}{\sum_{l\in\mathcal{J}} {\mathcal{P}_l^{0,k}e^{\delta_l+u_2(x_l^m,\nu_k,\alpha)}}}-1\bigg]=0,
		\end{equation*} 
		which is the moment condition (\ref{Moment: Fixed point first order condition}). 

	\end{proof}

	\section{ Proofs of Proposition \ref{prop: Consistency of theta}}
	\subsubsection*{Some Notations}
	Fix a $\theta$ and a persuasion strategy $\tilde{F}(s^{ID},\bm{\epsilon};\theta)$. Recall that I use $\hat{\mathcal{P}}_{j,s}^{0,k}(\bm{x}(l);\theta)$ to denote the estimated unconditional choice probability under persuasion signal $s$ solved from (\ref{eq: empirical alternative optimization}) and use $\hat{\mathbf{P}}_s(\theta)$ to denote the vector of all $j,k,l,s$. I use $\tilde{\mathcal{P}}_{j,s}^{0,k}(\bm{x}(l);\theta)$ to denote the true unconditional choice probability under persuasion solved from (\ref{eq: alterantive optimization}), and use $\tilde{\mathbf{P}}_s(\theta)$ to denote the vector of  all $j,k,l,s$. I use $\mathbf{P}^{0}$ to denote the true unconditioned choice probabilities without persuasion that corresponds to the moment condition (\ref{Moment condition: P}), and use $\hat{\mathbf{P}}^0$ to denote its estimator. I use $\mathbb{G}$ to denote the empirical distribution of $\hat{\epsilon}$ and use $G$ to denote the true distribution of $\epsilon$. I use $B_r(\cdot)$ to denote a neighborhood of radius $r$ near ($\cdot$).

	\subsection{Some Lemmas}
	\begin{assumption} \label{assumption: M-estimation of P}
		Fixing the index $k,l,s$, let  \[M(\{P_j\}_{j=1}^J,\theta)=\int_{\bm{\epsilon}} \sum_{j=1}^J P_j e^{\alpha_j^k(\bm{x}(l))+\epsilon_j} \tilde{F}^k(s|\bm{\epsilon};\theta)G(\bm{\epsilon}).\] 
		The following condition hold: $\forall \theta\in \Theta$, $\forall \kappa>0$, there exists some $\zeta>0$ such that 
		\[
		\sup_{d\left((P_j)_{j=1}^J,(\tilde{\mathcal{P}}_{j,s}^{0,k}(\bm{x}(l);\theta))_{j=1}^J\right)>\kappa}M(\{\tilde{\mathcal{P}}_{j,s}^{0,k}(\bm{x}(l);\theta)\}_{j=1}^J,\theta)-M(\{P_j\}_{j=1}^J,\theta)>\zeta.
		\]
	\end{assumption}

	\begin{lem}
		Fixing the index $k,l,s$. Let
		\[M_n(\{P_j\}_{j=1}^J,\theta)=\frac{1}{M(\bm{x}(l))}\sum_{m=1}^{M(\bm{x}(l))} \sum_{j=1}^J P_{j} e^{\alpha_{j,0}^k(\bm{x}(l))+\epsilon^m_j} \tilde{F}^k(s|\bm{\epsilon}^m;\theta)\mathbbm{1}(\mathbf{X}^m=\bm{x}(l)),\]
		where $M(\bm{x}(l))=\sum_{m=1}^M \mathbbm{1}(\mathbf{X}^m=\bm{x}(l))$. Suppose assumptions in Proposition \ref{prop: Consistency of theta} hold, then 
		\[\inf_{\theta\in\Theta}\bigg[M_n(\{\hat{\mathcal{P}}_{j,s}^{0,k}(\bm{x}(l);\theta)\}_{j=1}^J,\theta)-\sup_{(P_{j})_{j=1}^J\in \Delta^J} M_n(\{P_{j}\}_{j=1}^J,\theta)\bigg]=-o_p(1),\]
		where $\Delta^{J-1}$ is the $J$ dimensional probability simplex. 
	\end{lem}
	\begin{remark}
		The $M_n$ differs from the objective function of (\ref{eq: empirical alternative optimization}) because the $\alpha_{j,0}^k(\bm{x}(l))$ is the  true value of $\bm{\alpha}$, while we use $\hat{\bm{\alpha}}$ in  (\ref{eq: empirical alternative optimization}). This Lemma shows that $\hat{\mathbf{P}}_s(\theta)$ is also the $o_p(1)$-maximizer of $M_n$.
	\end{remark}

	\begin{proof}
		Define 
		\[\hat{M}_n(\{P_j\}_{j=1}^J,\theta)=\frac{1}{M(\bm{x}(l))}\sum_{m=1}^{M(\bm{x}(l))} \sum_{j=1}^J P_{j} e^{\hat{\alpha}_{j,0}^k(\bm{x}(l))+\hat{\epsilon}^m_j} \tilde{F}^k(s|\hat{\bm{\epsilon}}^m;\theta)\mathbbm{1}(\mathbf{X}^m=\bm{x}(l)),\]
		which is the objective function in (\ref{eq: empirical alternative optimization}), and $\{\hat{\mathcal{P}}_{j,s}^{0,k}(\bm{x}(l))\}_{j=1}^J$ is the maximizer of the above objective function in the simplex $\Delta^{J-1}$. Let $\{P_j^*(\theta)\}_{j=1}^J$ be the maximizer of $M_n(\{P_{j}\}_{j=1}^J,\theta)$, then we have 
		\begin{equation}\label{eq: append, op_1 maximizer}
		\begin{split}
		\hat{M}_n(\{\hat{\mathcal{P}}_{j,s}^{0,k}(\bm{x}(l))\}_{j=1}^J,\theta)
		&\ge \hat{M}_n(\{P_j^*(\theta)\}_{j=1}^J,\theta)\\
		&=\frac{1}{M(\bm{x}(l))}\sum_{m=1}^{M(\bm{x}(l))} \sum_{j=1}^J P^*_{j}(\theta) e^{\hat{\alpha}_{j,0}^k(\bm{x}(l))+\hat{\epsilon}^m_j} \tilde{F}^k(s|\hat{\bm{\epsilon}}^m;\theta)\mathbbm{1}(\mathbf{X}^m=\bm{x}(l))\\
		&= \frac{1}{M(\bm{x}(l))}\sum_{m=1}^{M(\bm{x}(l))} \sum_{j=1}^J P^*_{j}(\theta) f_j^m(1) \mathbbm{1}(\mathbf{X}^m=\bm{x}(l))\\
		\end{split}
		\end{equation}
		where the function $f_j^m$ is defined in the following:
		\[
		f_j^m(t)=e^{{\alpha}_{j,0}^k(\bm{x}(l))+t(\hat{\alpha}_{j,0}^k(\bm{x}(l))-{\alpha}_{j,0}^k(\bm{x}(l)))+{\epsilon}^m_j+t(\hat{\epsilon}^m_j-{\epsilon}^m_j)} \tilde{F}^k(s|{\bm{\epsilon}}^m+t(\hat{\bm{\epsilon}}^m-\bm{\epsilon}^m);\theta).
		\]
		By mean value theorem, we can find a $t_j^m\in[0,1]$ such that $f_j^m(1)=f_j^m(0)+(f_j^{m})^\prime(t_j^m)$. The derivatives with respect to $t$ is 
		\begin{equation}\label{eq: append, derivatives of f_j^m}
		\begin{split}
		(f_j^{m})^\prime(t)&= f_j^m(t)\left[\hat{\alpha}_{j,0}^k(\bm{x}(l))-{\alpha}_{j,0}^k(\bm{x}(l))+\hat{\epsilon}^m_j-{\epsilon}^m_j\right]\\
		&+ e^{{\alpha}_{j,0}^k(\bm{x}(l))+t(\hat{\alpha}_{j,0}^k(\bm{x}(l))-{\alpha}_{j,0}^k(\bm{x}(l)))+{\epsilon}^m_j+t(\hat{\epsilon}^m_j-{\epsilon}^m_j)}\sum_{i=1}^J\frac{\partial \tilde{F}^k}{\partial \epsilon_i}(\hat{\epsilon}_i^m-\epsilon_i^m).
		\end{split}
		\end{equation}
		Now I bound the term $\hat{\epsilon}^m_j-{\epsilon}^m_j$:
		\begin{equation}\label{eq: append, expansion of hat{epsilon}}
		\begin{split}
		|(\hat{\epsilon}^{m}_j-\epsilon_{j}^m)|&= \left|[\delta_j^{*}(\bm{ms}^m,D^m,\mathbf{X}^m,\hat{\bm{\alpha}},\hat{\mathbf{P}}^{0})-\delta_j^{*}(\bm{ms}^m,D^m,\mathbf{X}^m,{\bm{\alpha}},{\mathbf{P}}^{0})\right| \\
		&=\bigg|\sum_{j,k}\frac{\partial \delta_j^{*m}}{\partial \alpha_j^k(l)} (\hat{\alpha}_j^k(l)-{\alpha}_j^k(l))+ \sum_{j,k}\frac{\partial \delta_j^{*m}}{\partial \mathcal{P}_j^{0,k}(l)} (\hat{\mathcal{P}}_j^{0,k}(l)-{\mathcal{P}}_j^{0,k}(l)\bigg|\\
		&\le  J\times K \times C \max_{j,k}\left\{ \max \{\hat{\alpha}_j^k(l)-{\alpha}_j^k(l),\hat{\mathcal{P}}_j^{0,k}(l)-{\mathcal{P}}_j^{0,k}(l)\} \right\}
		\end{split}
		\end{equation}
		where the inequality holds by Assumption \ref{assumption: smooth parametric persuasion strategy}. Moreover, by Assumption \ref{assumption: smooth parametric persuasion strategy}, $\frac{\partial \tilde{F}^k}{\partial \epsilon_i} <C$ also holds. Now, denote the term $\max_{j,k}\left\{ \max \{\hat{\alpha}_j^k(l)-{\alpha}_j^k(l),\hat{\mathcal{P}}_j^{0,k}(l)-{\mathcal{P}}_j^{0,k}(l)\} \right\}$ by  $o^*_{\alpha,P}$, combining (\ref{eq: append, derivatives of f_j^m}) and (\ref{eq: append, expansion of hat{epsilon}}), we have 
		\[
		|(f_j^{m})^\prime(t)|\le 2J^2KC^2e^{{\alpha}_{j,0}^k(\bm{x}(l))+t(\hat{\alpha}_{j,0}^k(\bm{x}(l))-{\alpha}_{j,0}^k(\bm{x}(l)))+{\epsilon}^m_j+t(\hat{\epsilon}^m_j-{\epsilon}^m_j)}\times  o^*_{\alpha,P}.
		\]
		Now we substitute the mean value expansion of $f_j^m(1)$ back to (\ref{eq: append, op_1 maximizer}) to get 
		\begin{equation*}
		\begin{split}
		&\frac{1}{M(\bm{x}(l))}\sum_{m=1}^{M(\bm{x}(l))} \sum_{j=1}^J P^*_{j}(\theta) f_j^m(1) \mathbbm{1}(\mathbf{X}^m=\bm{x}(l))\\
		&= \frac{1}{M(\bm{x}(l))}\sum_{m=1}^{M(\bm{x}(l))} \sum_{j=1}^J P^*_{j}(\theta) f_j^m(0) \mathbbm{1}(\mathbf{X}^m=\bm{x}(l))\\
		&+ \frac{1}{M(\bm{x}(l))}\sum_{m=1}^{M(\bm{x}(l))} \sum_{j=1}^J P^*_{j}(\theta) (f_j^m)'(t_j^m) \mathbbm{1}(\mathbf{X}^m=\bm{x}(l))\\
		&\ge \frac{1}{M(\bm{x}(l))}\sum_{m=1}^{M(\bm{x}(l))} \sum_{j=1}^J P^*_{j}(\theta) f_j^m(0) \mathbbm{1}(\mathbf{X}^m=\bm{x}(l))\\
		&- 2J^2KC^2|o^*_{\alpha,P}| \frac{1}{M(\bm{x}(l))}\sum_{m=1}^{M(\bm{x}(l))} \sum_{j=1}^J P^*_{j}(\theta)e^{{\alpha}_{j,0}^k(\bm{x}(l))+t(\hat{\alpha}_{j,0}^k(\bm{x}(l))-{\alpha}_{j,0}^k(\bm{x}(l)))+{\epsilon}^m_j+t(\hat{\epsilon}^m_j-{\epsilon}^m_j)}\mathbbm{1}(\mathbf{X}^m=\bm{x}(l))
		\end{split}
		\end{equation*}
		By Lemma \ref{lemma: first stage normality}, $|o^*_{\alpha,P}|=o_p(1)$ and 
		\[
		\begin{split}
		&\quad \sum_{m=1}^{M(\bm{x}(l))} \sum_{j=1}^J P^*_{j}(\theta)e^{{\alpha}_{j,0}^k(\bm{x}(l))+t(\hat{\alpha}_{j,0}^k(\bm{x}(l))-{\alpha}_{j,0}^k(\bm{x}(l)))+{\epsilon}^m_j+t(\hat{\epsilon}^m_j-{\epsilon}^m_j)}\mathbbm{1}(\mathbf{X}^m=\bm{x}(l))\\
		&\rightarrow_p E\left[\sum_{j=1}^J P^*_{j}(\theta)e^{{\alpha}_{j,0}^k(\bm{x}(l))+{\epsilon}^m_j}\mathbbm{1}(\mathbf{X}^m=\bm{x}(l))\right]\le E\left[\sum_{j=1}^J e^{{\alpha}_{j,0}^k(\bm{x}(l))+{\epsilon}^m_j}\mathbbm{1}(\mathbf{X}^m=\bm{x}(l))\right],
		\end{split}
		\]
		where the last inequality holds because $P_j^*(\theta)\le 1$. The observation is that  \[E\left[\sum_{j=1}^J e^{{\alpha}_{j,0}^k(\bm{x}(l))+{\epsilon}^m_j}\mathbbm{1}(\mathbf{X}^m=\bm{x}(l))\right]\] is independent of the parameter $\theta$. 
		\begin{equation*}
		\begin{split}
		&\quad \frac{1}{M(\bm{x}(l))}\sum_{m=1}^{M(\bm{x}(l))} \sum_{j=1}^J P^*_{j}(\theta) f_j^m(1) \mathbbm{1}(\mathbf{X}^m=\bm{x}(l))\\
		&\ge \frac{1}{M(\bm{x}(l))}\sum_{m=1}^{M(\bm{x}(l))} \sum_{j=1}^J P^*_{j}(\theta) f_j^m(0) \mathbbm{1}(\mathbf{X}^m=\bm{x}(l))-o_p(1)\\
		&= \sup_{(P_{j})_{j=1}^J\in \Delta^J} M_n(\{P_{j}\}_{j=1}^J,\theta)-o_p(1),
		\end{split}
		\end{equation*}
		where the last equality holds by definition of $M_n(\{P_{j}\}_{j=1}^J,\theta)$ and $\{P^*_{j}(\theta)\}_{j=1}^J$ is the maximizer of $M_n(\{P_{j}\}_{j=1}^J,\theta)$. In particular, the $o_p(1)$ term $2J^2KC^2|o^*_{\alpha,P}| E\left[\sum_{j=1}^J e^{{\alpha}_{j,0}^k(\bm{x}(l))+{\epsilon}^m_j}\mathbbm{1}(\mathbf{X}^m=\bm{x}(l))\right]$ is independent of $\theta$, so the result in the Lemma follows. 
	\end{proof}
	
	\begin{lem}
		$\sup_{\theta\in\Theta, (P_{j})_{j=1}^J\in \Delta^{J}} |M_n(\{P_{j}\}_{j=1}^J,\theta)-M(\{P_{j}\}_{j=1}^J,\theta)|=o_p(1)$
	\end{lem}
	\begin{proof}
		
		Let $((P_{j})_{j=1}^J,\theta)$ and $((\bar{P}_{j})_{j=1}^J,\bar{\theta})$ be two values in the set $\Theta \times \Delta^{J}$. 
		\begin{equation*}
		\begin{split}
		&\left|\sum_{j=1}^J P_{j} e^{\alpha_{j,0}^k(\bm{x}(l))+\epsilon^m_j} \tilde{F}^k(s|\bm{\epsilon}^m;\theta)\mathbbm{1}(\mathbf{X}^m=\bm{x}(l))-\sum_{j=1}^J \bar{P}_{j} e^{\alpha_{j,0}^k(\bm{x}(l))+\epsilon^m_j} \tilde{F}^k(s|\bm{\epsilon}^m;\bar{\theta})\mathbbm{1}(\mathbf{X}^m=\bm{x}(l))\right|\\
		&\le_{(1)} \left(\sum_{i=1}^{dim(\theta)}(\bar{\theta}_i-\theta_i)\frac{\partial \tilde{F}^k}{\partial \theta_i} +\sup_j|P_j-\bar{P}_j| \right)\sum_{j=1}^J  e^{\alpha_{j,0}^k(\bm{x}(l))+\epsilon^m_j}\mathbbm{1}(\mathbf{X}^m=\bm{x}(l))\\
		&\le_{(2)} C\times dim(\theta) ||(\bar{P}_{j})_{j=1}^J,\bar{\theta})-(P_{j})_{j=1}^J,\theta)||_\infty \sum_{j=1}^J  e^{\alpha_{j,0}^k(\bm{x}(l))+\epsilon^m_j}\mathbbm{1}(\mathbf{X}^m=\bm{x}(l))\\
		&\le C\times C_1\times dim(\theta) ||(\bar{P}_{j})_{j=1}^J,\bar{\theta})-(P_{j})_{j=1}^J,\theta)||_\infty \sum_{j=1}^J  e^{\alpha_{j,0}^k(\bm{x}(l))+\epsilon^m_j}\mathbbm{1}(\mathbf{X}^m=\bm{x}(l))\\
		\end{split}
		\end{equation*}
		where $||\cdot||_\infty$ is the sup norm on a vector, and $C_1$ is a constant such that $||\cdot||_\infty\le ||\cdot||_2$ \footnote{Such $C_1$ can always be found because all norms of a finite dimensional vector space are equivalent.}. Inequality (1) follows by mean value theorem and inequality (2) follows by Assumption \ref{assumption: smooth parametric persuasion strategy}. Then by Theorem 2.7.11 in \citet{VDV}, we have the uniform convergence.	
	\end{proof}
	
	\begin{lem}
		If Condition \ref{assumption: M-estimation of P} holds, then $\sup_{\theta\in\Theta}|\hat{\mathbf{P}}_s(\theta)-\tilde{\mathbf{P}}_s^0(\theta)|=o_p(1)$
	\end{lem}
	
	\begin{proof}
		Lemma C.1 and C.2 implies that 
		\[\sup_{\theta\in\Theta}|M_n(\{\hat{\mathcal{P}}_{j,s}^{0,k}(\bm{x}(l))\}_{j=1}^J,\theta)-M(\{\tilde{\mathcal{P}}_{j,s}^{0,k}(\bm{x}(l))\}_{j=1}^J,\theta)|=o_p(1).\]
		So we have 
		\[
		\begin{split}
		&\quad\sup_\theta  M(\{\hat{\mathcal{P}}_{j,s}^{0,k}(\bm{x}(l))\}_{j=1}^J,\theta)-M(\{\hat{\mathcal{P}}_{j,s}^{0,k}(\bm{x}(l))\}_{j=1}^J,\theta)\\
		&\le\sup_\theta M(\{\hat{\mathcal{P}}_{j,s}^{0,k}(\bm{x}(l))\}_{j=1}^J,\theta)-M_n(\{\hat{\mathcal{P}}_{j,s}^{0,k}(\bm{x}(l))\}_{j=1}^J,\theta)+o_p(1)=o_p(1),
		\end{split}
		\]
		where the last equality hold by Lemma C.2. 
		By assumption C.1, the event \[d\left((P_j)_{j=1}^J,(\tilde{\mathcal{P}}_{j,s}^{0,k}(\bm{x}(l);\theta))_{j=1}^J\right)>\kappa\] is contained in the event $\sup_\theta  M(\{\hat{\mathcal{P}}_{j,s}^{0,k}(\bm{x}(l))\}_{j=1}^J,\theta)-M(\{\hat{\mathcal{P}}_{j,s}^{0,k}(\bm{x}(l))\}_{j=1}^J,\theta)>\kappa$, therefore
		{\footnotesize
			\[Pr\left(d\left((P_j)_{j=1}^J,(\tilde{\mathcal{P}}_{j,s}^{0,k}(\bm{x}(l);\theta))_{j=1}^J\right)>\kappa\right)<Pr(\sup_\theta  M(\{\hat{\mathcal{P}}_{j,s}^{0,k}(\bm{x}(l))\}_{j=1}^J,\theta)-M(\{\hat{\mathcal{P}}_{j,s}^{0,k}(\bm{x}(l))\}_{j=1}^J,\theta)>\kappa)\rightarrow 0\].}
		The result follows by taking the union over finite index $k=1,...,K$ and $l=1,...,L$. 
	\end{proof}
	
	\begin{lem}
		Let 
		$\mathbb{F}^k(s|\theta)\equiv 
		\frac{1}{M}\sum_{m=1}^M \tilde{F}^k(s|\hat{\bm{\epsilon}}^m;\theta)$ and let 
		$\tilde{F}^k(s|\theta)\equiv \int \tilde{F}^k(s|\bm{\epsilon};\theta)dG(\bm{\epsilon})$. The following hold under assumption \ref{assumption: smooth parametric persuasion strategy} $
		\sup_{\theta\in\Theta} |\tilde{\mathbb{F}}^k(s|\theta)-\tilde{F}(s|\theta)|=o_p(1)$.
	\end{lem}

	\begin{proof}
		We look at the following expansion 
		\begin{equation}\label{eq: append, expansion of F(s)}
		\begin{split}
		\left|\tilde{\mathbb{F}}^k(s|\theta)-\tilde{F}(s|\theta)\right|&=\frac{1}{M}\left|\sum_{m=1}^N \tilde{F}(s|\hat{\bm{\epsilon}}^m;\theta)-\tilde{F}(s|\bm{\epsilon}^m;\theta)+\tilde{F}(s|\bm{\epsilon}^m;\theta)-E_{\bm{\epsilon}}(\tilde{F}(s|\bm{\epsilon}^m;\theta))\right|\\
		&=\left|\frac{1}{M}\sum_{m=1}^M \sum_{j=1}^J\frac{\partial \tilde{F}}{\partial{\epsilon}_j}(\hat{\epsilon}^{m}_j-\epsilon^{m}_j)\right|+ \left|\frac{1}{M}\sum_{m=1}^M[ \tilde{F}(s|{\epsilon}_m;\theta)-E_{\epsilon}(\tilde{F}(s|{\epsilon}_m;\theta))]\right|\\
		&\le\frac{C}{M}\left|\sum_{m=1}^M \sum_{j=1}^J(\hat{\epsilon}^{m}_j-\epsilon^{m}_j)\right|+ \left|\frac{1}{M}\sum_{m=1}^M[ \tilde{F}(s|{\epsilon}_m;\theta)-E_{\epsilon}(\tilde{F}(s|{\epsilon}_m;\theta))]\right|,
		\end{split}
		\end{equation}
		where the last inequality holds by assumption \ref{assumption: smooth parametric persuasion strategy}. Now we use the expansion of $\epsilon_j^m$ in (\ref{eq: append, expansion of hat{epsilon}}) to get 
		\[
		\left|\frac{C}{M}\sum_{m=1}^M \sum_{j=1}^J(\hat{\epsilon}^{m}_j-\epsilon^{m}_j)\right|\le JKC^2\left|\max_{j,k}\left\{ \max \{\hat{\alpha}_j^k(l)-{\alpha}_j^k(l),\hat{\mathcal{P}}_j^{0,k}(l)-{\mathcal{P}}_j^{0,k}(l)\} \right\}\right|=o_p(1)
		\]
		Note that  $F(s|\epsilon_m;\theta)$ is a Donsker class indexed by $\theta$ by assumption \ref{assumption: smooth parametric persuasion strategy}, which implies 
		\[\sup_{\theta\in \Theta}\left|\frac{1}{M}\sum_{m=1}^M[ \tilde{F}(s|{\epsilon}_m;\theta)-E_{\epsilon}(\tilde{F}(s|{\epsilon}_m;\theta))]\right|=o_p(1).\]
		Combined the two terms in (\ref{eq: append, expansion of F(s)}) we can get $
		\sup_{\theta}|\tilde{\mathbb{F}}^k(s|\theta)-\tilde{F}^k(s|\theta))=o_p(1)$. 
	\end{proof}
	
	\begin{lem}
		$\sup_{\theta\in\Theta,j=1,...J, k=1,...K,l=1,...L} |\hat{\mathcal{P}}_{j,s}^{0,k}(\bm{x}(l);\theta)\tilde{\mathbb{F}}^k(s|\theta)-\tilde{\mathcal{P}}_{j,s}^{0,k}(\bm{x}(l);\theta) \tilde{F}^k(s|\theta)|=o_p(1)$.
	\end{lem}
	
	\begin{proof}
		This follows directly from Lemma C.3 and C.4.
	\end{proof}
	
	\begin{lem}\label{lem: append, o_p(1) minimizer of theta_hat}
		Consider 
		\begin{equation}\label{eq:append, moment function of theta with true }
		g^*_{l,j,k}(\theta,\tilde{\mathbf{ms}}^m,D^m,\mathbf{X}^m,\tilde{\mathbf{P}}_s)= [\tilde{ms}_j^m-\sum_{d=1}^K h_j^{*k}(\theta,\tilde{\mathbf{P}}_s,\bm{x}(l))d_k^m] d_k^m\mathbbm{1}(\mathbf{X}^m=\bm{x}(l)), 
		\end{equation}
		\begin{equation} \label{equation: append, h-estimator with}
		h_j^{*k}(\theta,\tilde{\mathbf{P}}_s,\bm{x}(l))=\sum_{s} \bigg[\tilde{\mathcal{P}}_{j,s}^{0,k}(\bm{x}(l),\theta) \tilde{F}^k(s|\theta) \bigg].
		\end{equation}
		The equation (\ref{eq:append, moment function of theta with true }) and (\ref{equation: append, h-estimator with}) differ from (\ref{eq: moment function of theta with plug in estimator}) and (\ref{equation: h-estimator}) because (\ref{eq:append, moment function of theta with true }) and (\ref{equation: append, h-estimator with}) use the true unconditional choice probability instead of the estimator.  Define \[
		L_n(\theta)=\left(\frac{1}{N}\sum_{m=1}^N\mathbf{g}^*(\theta)\right)^\prime W \left(\frac{1}{N}\sum_{m=1}^N\mathbf{g}^{*}(\theta)\right),\] where $ \mathbf{g}^*(\theta)$ collects $g_{l,j}^*$ for all $l,j,k$ indices. 
		Then $\hat{\theta}$ is an $o_p(1)$ minimizer of $L_n(\theta)$, i.e. $
		L_n(\hat{\theta})\le \min_\theta L_n(\theta) +o_p(1)$. 	
	\end{lem}
	\begin{proof}
		Note that $\hat{\theta}=\arg\min \hat{L}_n(\theta)$, where $\hat{L}_{n}(\theta)$ is the objective function of (\ref{eq: GMM objectives theta}).
		
		I first denote 
		$\Delta_{l,j,k}(\theta)= g^*_{l,j,k}(\theta,\tilde{\mathbf{ms}}^m,D^m,\mathbf{X}^m,\tilde{\mathbf{P}}_s)- g_{l,j,k}(\theta,\tilde{\mathbf{ms}}^m,D^m,\mathbf{X}^m,\hat{\mathbf{P}}_s)$, where $g_{l,j}$ is defined in (\ref{eq: moment function of theta with plug in estimator}). Using the expression of $g^*_{l,j,k}$ and $g_{l,j,k}$, we have  
		\begin{equation}\label{eq: append, op(1) of Delta}
		\sup_{\theta\in \Theta}|\Delta_{l,j,k}(\theta)|\le \sup_{\theta\in \Theta}\left|\sum_{d=1}^K\sum_s\left( \hat{\mathcal{P}}_{j,s}^{0,k}(\bm{x}(l);\theta)\tilde{\mathbb{F}}^k(s|\theta)-\tilde{\mathcal{P}}_{j,s}^{0,k}(\bm{x}(l);\theta) \tilde{F}^k(s|\theta)\right)\mathbbm{1}(\mathbf{X}^m=\bm{x}(l))\right| =o_p(1).
		\end{equation}
		The difference  $L_n(\theta)-\hat{L}_n(\theta)=\bm{\Delta}(\theta)'W\bm{\Delta}(\theta)$, where $\bm{\Delta}(\theta)=(\Delta_{l,j}(\theta))_{l,j}$. Then by (\ref {eq: append, op(1) of Delta}), $\sup_\theta |L_n(\theta)-\hat{L}_n(\theta)|\le ||\bm{\Delta}(\theta)||_2^2 \max eig(W)=o_p(1)$.

		Now  I look at $L_n(\hat{\theta})$. Suppose we can find $\theta^*$ such that  $L_n(\theta^*)\le \inf_{\theta\in\Theta} L_n(\theta)+o_p(1)$ 
		\begin{align*}
		L_n(\hat{\theta})&=\hat{L}_n(\hat{\theta})+L_n(\hat{\theta})-\hat{L}_n(\hat{\theta})\\
		&\le_{(1)} \hat{L}_n({\theta}^*)+L_n(\hat{\theta})-\hat{L}_n(\hat{\theta})\\
		&= {L}_n({\theta}^*)+ \underbrace{L_n(\hat{\theta})-\hat{L}_n(\hat{\theta})}_{o_p(1)}
		-[\underbrace{L_n({\theta}^*)-\hat{L}_n({\theta}^*)}_{o_p(1)}]\\
		&=_{(2)} L_n(\theta^*) +o_p(1)\le \inf_{\theta\in\Theta} L_n(\theta)+o_p(1)
		\end{align*}
		where inequality $(1)$ holds by the definition of $\hat{\theta}$, and $(2)$ equality holds because we have shown $\sup_\theta |L_n(\theta)-\hat{L}_n(\theta)|\le ||\bm{\Delta}(\theta)||_2^2 \max eig(W)=o_p(1)$. 
	\end{proof}
	
	\begin{lem}\label{lem: append, uniform convergence of L_n to L}
		Let $L(\theta)=E[\mathbf{g}^*(\theta)]'WE[\mathbf{g}^*(\theta)]$. 
		Then $\sup_{\theta\in\Theta} |L_n(\theta)-L(\theta)|=o_p(1)$
	\end{lem}
	\begin{proof}
		Define the difference  
		\[
		\begin{split}
		\Delta_{l,j,k}^*(\theta)&= \frac{1}{N}\sum_{m=1}^N \left(\tilde{ms}_j^m\mathbbm{1}(\mathbf{X}^m=\bm{x}(l))d_k^m-E[\tilde{ms}_j^m\mathbbm{1}(\mathbf{X}^m=\bm{x}(l))d_k^m]\right)+\\
		&+ \sum_{k'=1}^K\left(\frac{1}{N}\sum_{m=1}^N d_{k'}^m \mathbbm{1}(\mathbf{X}^m=\bm{x}(l))d_k^m-E[d_{k'}^m\mathbbm{1}(\mathbf{X}^m=\bm{x}(l))d_k^m]\right)\mathcal{P}_{1,s}^{0,k}(\theta)\tilde{F}^k(s|\theta)
		\end{split}
		\]
		Observe that $\mathcal{P}_{1,s}^{0,k}(\theta)\tilde{F}^k(s|\theta)\in [0,1]$ because it is the product of two probability quantities. Moreover, the $d_k^m\in[0,1]$. Therefore, we can bound $\bm{\Delta}^*(\theta)$, which is the vector of $\Delta_{l,j,k}^*$ for all $l,j,k$ indices by 
		\begin{equation}\label{eq: append, delta^*}
		||\bm{\Delta}^*(\theta)||_2\le  JK L  \left|\frac{1}{N}\sum_{m=1}^N \tilde{ms}_j^m-E[\tilde{ms}_j^m]\right|+JK^2L \max_{k,k'} \left|\frac{1}{N}\sum_{m=1}^N d_k^m d_{k'}^m-E[d_k^md_{k'}^m]\right|.
		\end{equation}
		The right hand side of (\ref{eq: append, delta^*}) does not depend on $\theta$. By apply weak law of large numbers to the sample means of $\tilde{ms}_j^m$ and $d_k^md_{k'}^m$, we have $\sup_{\theta\in \Theta} ||\bm{\Delta}^*(\theta)||_2= o_p(1)$.
		Then notice that $L(\theta)-L_n(\theta)=\bm{\Delta}^*(\theta)'W\bm{\Delta}^*(\theta)$, so we have 
		\[\sup_{\theta\in\Theta}|L(\theta)-L_n(\theta)|\le \max eig(W) ||\bm{\Delta}^*(\theta)||_2^2=o_p(1). \]
	\end{proof}
	
	\subsection{Main Proof of Proposition \ref{prop: Consistency of theta}}
	\begin{proof}
		The consistency of $\hat{\theta}$ follows by the identification assumption 
		\[\sup_{d(\theta,\theta_0)>\zeta} L(\theta)-L(\theta_0)>0\]
		where $\hat{\theta}$ is an $o_p(1)$ minimizer of $L_n$ by Lemma  \ref{lem: append, o_p(1) minimizer of theta_hat}. Moreover  we have the uniform convergence of $\sup_{\theta\in\Theta}|L_n(\theta)- L(\theta)|=o_p(1)$ by Lemma \ref{lem: append, uniform convergence of L_n to L}. Conditions of Theorem 5.7 in \citet{van2000asymptotic} are satisfied, so $\hat{\theta}\rightarrow_p \theta$. 
	\end{proof}

	\section{Discussion of Computation} \label{Discussion of Computation}
	The estimators in the main text are constructed in two steps. While the joint estimation of $(\alpha,\beta,\mathbf{P},\theta)$ is possible, the computational burden is heavy.  Markets with persuasion also provide identification power to the first stage parameter $(\alpha,\beta,\mathbf{P})$, but this requires me to use contraction mapping each time I search over a higher dimensional parameter space when including $\theta$. Also, the estimation of $\theta$ requires solving the empirical optimization problem (\ref{eq: empirical alternative optimization}) for given first stage parameters. For the two-step estimation, I just plug in the first stage estimator and solve the (\ref{eq: empirical alternative optimization}) for different values of $\theta$, while for joint estimation the optimization problem needs to be repeated for each guessed value of $(\alpha,\beta,\mathbf{P})$. 
	
	The computational burden also comes from the contraction mapping because I need to iterate over $M$ markets. So here I use the following trick to convert the $M$ contraction mappings to one single contraction mapping.
	
	\begin{prop}
		Let $T^m(\delta):R^d\rightarrow R^d$ be a contraction mapping for each $m=1...M$. Then $T\equiv (T^1,...,T^M):R^{dM}
		\rightarrow R^{dM} $ is a contraction mapping acting on $(\delta^1,...,\delta^M)$.
	\end{prop}
	\begin{proof}
		Let $C^m<1$ be the contraction constant that $|T^m(\delta_1)-T^m(\delta_2)|<C^m|\delta_1-\delta_2|$. Then  $C(M)=\max{C^m}<1$ is the contraction constant for $T$. 
	\end{proof}
	The computational burden for this combined $T$ can be potentially high because: 1. even though iteration on matrix is faster than iteration over $M$ markets, the iteration on matrix is still slow when $M$ is large; 2. the uniform contraction constant $C(M)$ can be close to one and number of iteration to achieve certain tolerance level may be large. The following algorithm is helpful to reduce the running time:
	\begin{itemize}
		\item Set up a tolerance level $tol$ and  a threshold integer $K_{thr}$. Run the iteration on $T$, and count the number of $\delta^m$ such that $|T(\delta^m)-\delta^m|<tol$, denote this number as $K_{con}$.
		\item  When $K_{con}>K_{thr}$, collect the remaining markets index and construct the new contraction $T'=\{T^m\}_{remain}$. Iterate until convergence.
		\item Multiple threshold  to decide remaining markets can be set up to further boost the speed.
	\end{itemize}
	This algorithm exploits the fact that the contraction mapping is simply the stacking of individuals. The intuition is that if $C^m\in \{C^{small},C^{large}\}$, and when the number of markets falls into $C^{large}$ group is relatively small, the algorithm avoids the slow iteration on $C^{large}$ and also avoid excessive iteration on markets with $C^{small}$. This algorithm boost the computation speed of contraction mapping by a factor of 10 in the empirical application.

\end{document}